\documentclass[journal,compsoc]{IEEEtran}
\usepackage{graphicx}
\usepackage{graphics}
\usepackage{subfigure}
\usepackage{multirow}
\usepackage{textcomp}
\usepackage{array}

\usepackage{algorithm}
\usepackage{algorithmic}
\usepackage[algo2e,linesnumbered,ruled]{algorithm2e}

\usepackage{amsthm}

\newtheorem{definition}{Definition}

\newtheorem{proposition}{Proposition}
\newtheorem{thm}{Theorem}

\usepackage{amsmath}
\usepackage{amssymb}

\usepackage[cmintegrals]{newtxmath}

\newcolumntype{M}[1]{>{\centering\arraybackslash}m{#1}}

\ifCLASSOPTIONcompsoc
\usepackage[nocompress]{cite}
\else
\usepackage{cite}
\fi

\ifCLASSINFOpdf
\else
\fi
\begin{document}

\title{The Coverage Overlapping Problem of Serving Arbitrary Crowds in 3D Drone Cellular Networks}


\author{Chuan-Chi Lai,~\IEEEmembership{Member,~IEEE},
	Li-Chun Wang,~\IEEEmembership{Fellow,~IEEE}, and Zhu Han,~\IEEEmembership{Fellow,~IEEE}
\IEEEcompsocitemizethanks{
	\IEEEcompsocthanksitem This work was supported by Ministry of Science and Technology under the Grant No. MOST 108-2634-F-009-006- and MOST 109-2634-F-009-018- through Pervasive Artificial Intelligence Research (PAIR) Labs, Taiwan, and partially supported by the Higher Education Sprout Project of the National Chiao Tung University and Ministry of Education, Taiwan. This work was also partially supported by US NSF EARS-1839818, CNS-1717454, CNS-1731424, and CNS-1702850.
	\IEEEcompsocthanksitem C.-C. Lai and L.-C. Wang are with the Department of Electrical and Computer Engineering, National Chiao Tung University, 300 Hsinchu, Taiwan. (Corresponding author's e-mail: lichun@g2.nctu.edu.tw)
	\IEEEcompsocthanksitem Z. Han is with the University of Houston, Houston, TX 77004 USA, and also with the Department of Computer Science and Engineering, Kyung Hee University, Seoul, South Korea, 446-701.
	
}
}

\markboth{Preprint for IEEE Transactions on Mobile Computing 
}%
{Lai \MakeLowercase{\textit{et al.}}
}
%



\IEEEtitleabstractindextext{%
\begin{abstract}
Providing coverage for flash crowds is an important application for drone base stations (DBSs). However, any arbitrary crowd is likely to be distributed at a high density. Under the condition for each DBS to serve the same number of ground users, multiple DBSs may be placed at the same horizontal location but different altitudes and will cause severe co-channel interference, to which we refer as the coverage overlapping problem. To solve this problem, we then proposed the data-driven 3D placement (DDP) and the enhanced DDP (eDDP) algorithms. The proposed DDP and eDDP can effectively find the appropriate number, altitude, location, and coverage of DBSs in the serving area in polynomial time to maximize the system sum rate and guarantee the minimum data rate requirement of the user equipment. The simulation results show that, compared with the balanced $k$-means approach, the proposed eDDP can increase the system sum rate by 200\% and reduce the computation time by 50\%. In particular, eDDP can effectively reduce the occurrence of the coverage overlapping problem and then outperform DDP by about 100\% in terms of system sum rate.
\end{abstract}

}

\maketitle

\IEEEdisplaynontitleabstractindextext

%
\IEEEpeerreviewmaketitle

\IEEEraisesectionheading{\section{Introduction}\label{sec:introduction}}
\IEEEPARstart{D}{ue} to the popularity of the Internet of Things (IoT), the demand for mobile data traffic for the upcoming 5G and beyond 5G wireless networks is growing rapidly. According to the latest report~\cite{7452276}, the increase in demand is similar to the memory and computing power growth following Moore's law in the last 30 years. With this trend, The expected peak wireless data rates will increase to about 10 Gb/s with 5G in 2020 and the global mobile data traffic will reach 1 zettabyte/mo until 2028. This will lead the telecom operators facing the great capacity demands and needs a heavy burden on operational costs and infrastructure updates. 

To meet these growing demands, some early works~\cite{6812286}~\cite{Chai2018}~\cite{8411547}~\cite{8369148} have been dedicated to self-organizing networks (SONs) and heterogeneous networks (HetNets) (i.e., deploy various small cells).
However, solving this problem by deploying various ground base stations (GBSs), such as ultra-dense small cells (UDSC) or Wi-Fi access points, on the ground lacks the flexibility for dynamic, unexpected or emergency situations like outdoor concerts, election campaigns, disaster relief, and so forth. If the cellular operators wants to provide good services in the above cases, they have to pay an economically infeasible budget for the deploying cost of UDSC and the after-cost of its management. Drone base stations (DBSs) thereby become a new promising solution for providing temporary communication services to recover the disaster area or to satisfy the sudden demands (hot-spots) caused by flash crowds, which is commonly referred to as drone-assisted communications~\cite{9024679}~\cite{9048611}~\cite{7744808}~\cite{8038869}~\cite{8316776} or drone-enabled offloading services~\cite{8740949}. 

The advantage of using DBSs is their ability to provide on-the-fly communications. The DBSs can effectively establish line-of-sight (LoS) communication links and mitigate signal blockage and shadowing since they fly at relatively high altitude. Although the spatio-temporal distribution of user mobility~\cite{Ernest2015}~\cite{8673556} and mobile data traffic~\cite{8117559}~\cite{XU2018146} are arbitrary and hard to predict, which makes the design of future cellular systems an tremendous challenge, DBSs can move towards the potential ground users and establish reliable connections with a low transmit power due to the flexibility of their altitude and mobility. Accordingly, DBSs become an agile solution to serve ground users arbitrarily spread over a geographical area with the limited terrestrial infrastructure. Compared to the placement of traditional GBSs, deploying DBSs is a cost-effective and energy-efficient solution which can save a large amount of land cost for the cellular operators. 

However, most existing works consider drones as cellular-connected aerial user equipments (UEs) that must connect to a wireless network so as to operate. The cellular-connected aerial UEs are often used for wide range of applications such as surveillance~\cite{8255734}, remote sensing~\cite{8572727}~\cite{Hu2018apwcs}, virtual reality~\cite{7849534}, and package delivery~\cite{8690828}. In fact, the above different applications are the variants of IoT applications. Obviously, wireless communications technologies play the most common and important role in these applications. To effectively use drones for serving IoT (or massive users), several technical challenges must be addressed such as optimal deployment, mobility and energy-efficient use of drones as outlined in~\cite{7412759}. Compared to traditional GBSs, DBSs are more flexible and can be easily deployed at some specific positions for serving the dynamic crowds and sudden events, such as outdoor concert, baseball game, and festival crowds. 

In practice, there are many challenging open issues for establishing such a drone-assisted cellular system. In particular, the managing mechanisms of placement, resource allocation, power control~\cite{8865486}, and flight scheduling are the urgent technologies to allow the deployed DBSs to coexist with the terrestrial cellular systems. Such a drone-assisted cellular system will be one important use case of 5G or beyond 5G networks, which is capable to serve dynamic traffic demands~\cite{7762185}. The effective and efficient technologies of DBS placement/management thereby become hot topics in communications domain. Hence, we focus on the dynamic placement of DBSs over a terrestrial cellular system while improving the system performance in terms of sum rate and co-channel interference (overlapping area)~\cite{globe19_Lai}.

In this work, we discuss how to deploy multiple DBS to the appropriate height and location to serve the ground UEs. We model this problem as a maximizing system sum rate problem, taking into account the predetermined minimum data rate requirement of UE and co-channel interference between different DBSs.
Focusing on the downlink transmission from the GBS to users and from the DBSs to their corresponding users, we propose a \emph{data-driven 3D placement} (DDP) algorithm which solves the considered placement problem in a three-stage iterative searching framework. 
In the first stage, the proposed method determines the minimum (or lower-bound) possible number of deployed DBS based on the input information, system parameters, and optimization constraints.
In the second stage, the proposed method analyzes the distribution and density of ground UEs, and then finds the possible candidate placements for clustered UEs. In the last stage, the proposed method re-tunes the candidate altitude, location, and coverage of each DBS in the considered area for maximizing the system sum rate with the optimization constraints on the co-channel interference and the allocated data rate of each UE. The system will redo the above three-stage process with the increasing number of DBSs until all the optimization constraints are satisfied.

However, we discover a new issue, \emph{coverage overlapping problem}, which may occur when the system uses a $k$-means based clustering approach to place multiple DBSs to collaborate with the GBS to serve an arbitrarily distributed crowds. Such a special issue makes the system place multiple DBSs at similar locations with different altitudes and thus leads huge co-channel interference. An example of the coverage overlapping problem is depicted in Fig.~\ref{fig:ex_overlapping_placement}. We then add a pre-partition process into the proposed framework and this process will be executed right after the initialization stage. The pre-partition process uses the location of the GBS to divide the considered area into multiple sub-regions. The system then perform the operations of user association clustering and placement refinement for the partitioned sub-regions, respectively. The DDP with the above enhancement operations are called as \emph{enhanced data-driven 3D placement} (eDDP). The proposed eDDP can effectively reduce the occurrence of the coverage overlapping problem. 
In the simulation, we compare the placement of proposed DDP and eDDP with the result obtained by the conventional balanced $k$-means approach and then observe how large the overlapping area representing channel interference is.
We also observe the impacts of different numbers of UE and the pre-defined target satisfaction rate on the performance of different placement approaches in terms of computation time, number of deployed DBSs, and system sum rate. 
\begin{figure}[!t]
	\centering
	\includegraphics[width=1\columnwidth]{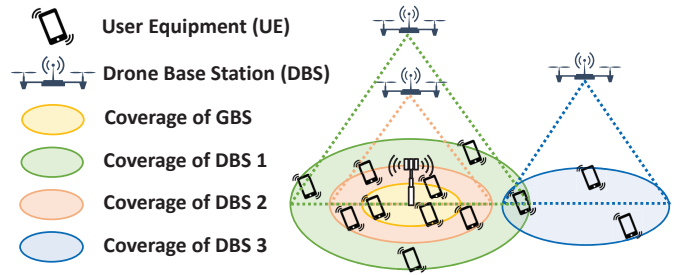}
	\caption{Example of the coverage overlapping problem.}
	\label{fig:ex_overlapping_placement}
\end{figure}

The main contributions of this paper are summarized as follows:
\begin{itemize}
	\item We prove that the placement problem of DBSs coexisting with the cellular system is NP-hard, and propose a data-driven approach to provide a sub-optimal placement for serving arbitrary crowds.			
	\item The proposed DDP models the placement refinement as a minimum closing circle problem~\cite{10.1007/BFb0038202} and effectively maximizes the system sum rate with the minimized co-channel interference with a linear time procedure.
	\item We also identify a new issue, coverage overlapping problem, which occurs while deploying multiple DBSs at different altitudes with different corresponding coverage radii, and then propose an enhanced approach, eDDP, to solve this issue.
	\item Unlike the conventional unsupervised learning approach, balanced $k$-means~\cite{10.1007/978-3-662-44415-3_4}, the proposed DDP and eDDP are adaptive to the scenarios of different scales and can automatically determine the appropriate number of DBSs to serve the arbitrarily distributed UEs. 
	\item The simulation results show that eDDP is the best approach. Compared to balanced $k$-means approach, it can increase the system sum rate by 200\% and reduce computation time by 50\%.  	
\end{itemize}

The rest of this paper is organized as follows. In Section~\ref{sec:related:works}, we review the state of the art and give a comparison summary. Section~\ref{sec:problem} presents the considered system model, assumptions, and problem statement of this work. Section~\ref{sec:3Dplacement} introduces the proposed data-driven approach and a breakdown of the algorithms. A new challenge appears while applying the $k$-means based approach and the proposed enhancement will be addressed in Section~\ref{sec:eddp}. The time complexity will also be analyzed in Section~\ref{sec:analysis}. Simulation results and comparison summary are presented in Section~\ref{sec:simulation}. Finally, we make concluding remarks in Section~\ref{sec:conclusion}.
\begin{table*}[!t]
	\renewcommand{\arraystretch}{1.2}
	\caption{Comparisons of Related Works and the Proposed Method}
	\label{compared_methods}
	\centering
	\begin{tabular}{M{.8cm}M{5.5cm}M{1.4cm}M{1.4cm}M{1.4cm}M{1.7cm}M{1.1cm}M{1.5cm}}
		\hline
		Method & Objective & Number of Drones & Altitudes of Drones & User Distribution & GBS Coexistence & Sumrate & Co-Channel Interference Awareness\\	\hline
		\hline
		\cite{6863654} &\raggedright Model the relation between the optimal DBS altitude and its coverage area & Single & -- & -- & No & No & -- \\ \hline
		\cite{7510820} &\raggedright Find the optimal altitude, coverage area, and location of DBS to maximized the number of covered users & Single & -- & Uniform & No & No & -- \\ \hline
		\cite{7962642} & \raggedright Optimize the DBS altitude and location with back-haul constraint & Single & -- & PCP & No & No & -- \\ \hline
		\cite{8642333} & \raggedright Find the appropriate altitude and location of a DBS to guarantee a minimum data rate requirement to ground users & Single & -- & Arbitrary & No & Yes & -- \\ \hline
		\cite{8269064} & \raggedright Maximize the energy efficiency of deployed DBSs in multiple pre-partitioned subregions & Single DBS for each subregion & Various & PPP & No & No & No \\ \hline
		\cite{7762053} & \raggedright Maximize the number of covered users with the minimum number of DBSs & Various & Fixed & Uniform & No & No & No \\	\hline
		\cite{7881122} & \raggedright Minimize the number of DBSs to cover all the users & Various & Various & Uniform with two densities & No & No & Yes \\ \hline
		\cite{8758183} & \raggedright Maximize the number of covered users continuously with a heuristic re-position algorithm & A fixed number of DBSs & Fixed & Dense PPP & GBS and DBS use different bands & No & Yes \\ \hline
		\cite{7756327} & \raggedright Determine the optimal density of DBSs from the spectrum sharing aspect & Various & Various & PPP/PHP & Yes & Yes & Yes \\ \hline
		\cite{7875131} & \raggedright Apply mobility prediction to drone cache for optimize the users' quality of experience and minimize the transmit power & A fixed number of DBSs & Various & Uniform & No & No & Yes \\ \hline
		\textbf{This paper} & \raggedright \textbf{Determine the appropriate number, positions, and altitudes of deployed DBSs automatically for satisfying the dynamic traffic demands (arbitrarily distributed users) with a maximized the system sum rate while coexisting with a GBS} & \textbf{Various} & \textbf{Various} & \textbf{Arbitrary} & \textbf{Yes} & \textbf{Yes} & \textbf{Yes} \\ \hline
	\end{tabular}
\end{table*}
%


\section{Related Work}
\label{sec:related:works}
The DBS (or relay) placement problem has recently attracted great attentions in literature, where various methods are proposed towards different considered objectives and requirements. We focus on the issues of drone-assisted wireless communications, we thereby survey and summarize the some related works in Table~\ref{compared_methods} with detail as follows.

Due to the characteristics of wireless propagation, there is a relation between the altitude and provided optimal coverage of a DBS, which is firstly modeled in~\cite{6863654}. The authors proposed an air-to-ground (ATG) channel model with derivations of the probabilities of line-of-sight (LoS) and non-line-of-sight (NLoS) links, and now their proposed channel model has been widely used in drone communication. In consideration of the path loss constraint and uniform users in different environments, the optimal altitude, coverage, and location single deployed DBS are discussed in~\cite{7510820}. 
A backhaul-aware robust 3D placement~\cite{7962642} of a DBS was proposed for temporarily increasing the network capacity or coverage of an area in 5G+ environments. The authors considered the case of deploying single DBS in the urban area and discussed the performance with constraints on the total available bandwidth and aggregate peak rate of the DBS's back-haul link. A heuristic method to deploy a single DBS for serving arbitrarily distributed users in polynomial time has been discussed in~\cite{8642333}. This method can provide guaranteed data rates to users under the consideration of the backhaul constraint.

After discussing the case of deploying single DBS, many researchers have moved their eyes on the issues of 3D placement of multiple DBSs. In~\cite{8269064}, the author divided the service area into multiple sub-regions with different densities of users and then considered the power consumption of the on-board circuit and communication to find the optimal heights of the DBSs deployed in these sub-regions. A spiral placement algorithm was proposed by~\cite{7762053} and it deployed multiple DBSs at better locations and minimize the number of DBSs to covered all users while considering various user densities. However, this approach only considers the fixed coverage and altitude of DBS for the placement. Another work~\cite{7881122} focused on optimizing the number of deployed DBSs and proposed a method based on particle swarm optimization~\cite{488968} for realizing the average behavior of target users over a 2D area, as well as deploying the minimal number of DBSs with recommended altitudes and locations. 
In~\cite{8758183}, a heuristic re-position algorithm was proposed to continuously maximize the number of covered users with a fixed number of DBSs while considering the existing ground small cells and inter-drone co-channel interference.

However, all the above existing works only considered research issues from the perspectives of users and did not consider the coexistence of ground cellular systems. From the cellular operator's aspect, \cite{7756327} used the 2D Poison Point Process (PPP) to generate the distribution of users in the serving area and also used the 3D Poisson Hole Process (PHP)~\cite{7557010} to distribute the DBSs in the air. The authors discussed the spectrum-throughput efficiency problem in such a drone-assisted cellular system and also considered the effect of co-channel interference between drone-cells and the GBS. \cite{8876665} considered that all users, UAVs and GBSs are in independent homogeneous PPP (HPPP) distribution, optimized the volume spectrum efficiency and multi-cell coverage probability of a multiple-input-single-output (MISO) mmWave UAV network.

From the system level aspect, some modern technologies can help the drone-assisted cellular system to improve the performance, such as cloud radio access networks (CRANs)~\cite{7444125} and edge computing~\cite{8436041}.
Hence, each DBS can be treated as an intelligent edge relay node and equipped with storage to cache popular data. An cloud-assisted infrastructure was considered in~\cite{7875131} and the authors apply a machine learning framework, Echo State Networks (ESN)~\cite{7880663}, to the considered environment. They used ESN to predict the user mobility pattern and the behavior of data access so that each drone can cache popular data. In this way, the served users will have a high probability to successfully access their required contents without using the back-haul connections of DBSs. Thus, the transmit power of each DBS and the QoE of each user can be improved.

Unfortunately, all the above related works did not consider the arbitrary distribution of users for flash-crowd events, such as outdoor concerts, marathons, 
election campaigns. They only observed the system performance using the statistical-based user distributions. In practice, the user distribution is non-uniform and it is difficult to model precisely with mathematical formulas. Furthermore, most of them also did not consider the coexistence of GBSs, the effects of co-channel interference, and the system sum rate optimization problem. The advantages of our work are summarized at the bottom of Table~\ref{compared_methods}.

\section{Problem Description}
\label{sec:problem}
To improve the readability, the main variables used throughout the paper are given in Table~\ref{Notations}. Other temporarily used variables will be explained in an in-text manner.

\subsection{System Model}
\label{sec:systemmodel}
As shown in Fig.~\ref{fig:system}, we consider a drone-assisted cellular system consisting one GBS, $G$, and a set of DBSs, $\mathcal{U}=\{U_1, U_2, \dots, U_K\}$, in an urban scenario, where $K$ is the maximum number of available DBSs. The drone-assisted cellular system serves a set of UEs, $E=\{u_1,u_2,\dots,u_N\}$, and the total number of UEs is $|E|=N$. The DBSs can move in the sky to any position. Each UE only uses the resource of one BS (GBS or DBS) at a certain time. We assume that all the UEs are arbitrarily distributed on the ground due to the operation requirements, the terrain limitations, or unpredictable events. The placement decision of DBSs is controlled by the edge controller (or controller) behind the GBS. All the DBS and UEs are equipped with directional antennas to transmit and receive 4G LTE-A signals in the considered environments. We assume that the GBS and each DBS use the same spectrum and provide the same bandwidth $B$ for the down links in the considered system. The GBS is also equipped a mmWave directional antenna array using another dedicated spectrum to provide an additional network volume for the back-haul communication link between the GBS and each DBS. Some similar system models/architectures using different bands for GBS-to-DBS and DBS-to-UE links already have been discussed in~\cite{8758183}~\cite{8761403}.

In our work, we focus on the downlink transmissions, and we introduce the radio propagation models for the downlink transmissions which consists of following three cases: 1) GBS to DBS, 2) DBS to UE, and 3) GBS to UE. We now respectively introduce these cases under the assumption that the appropriate number of deployed DBSs is $k$, where $1\leq k\leq K$.

\begin{table}[!t]
	\renewcommand{\arraystretch}{1.1}
	\caption{The Main Variables Used Throughout This Paper}
	\label{Notations}
	\centering
	\begin{tabular}{|c|p{6.8cm}|}
		\hline
		\textbf{Variable} & \textbf{Description} \\
		\hline 
		$\mathcal{A}$ & The target rectangle area \\	
		\hline
		$N$ & Number of UEs \\	
		\hline
		$k$ & Number of DBSs \\	
		\hline
		$K$ & Maximum number of DBSs  \\	
		\hline
		$P_G$ & Cellular transmit power of the GBS \\	
		\hline
		$P_G^{\text{bk}}$ & mmWave transmit power of the GBS  \\	
		\hline
		$P_j$ & Transmit power of each DBS  \\	
		\hline
		$L_{\text{allowable}}$ & Maximum allowable path-loss of GBS-to-DBS and DBS-to-UE links \\	
		\hline
		$\alpha$ & Path-loss exponent of the GBS-to-UE downlink \\	
		\hline
		$c_{\min}$ & Minimum data rate requirement of each UE \\	
		\hline
		$\tau$ & Target satisfaction rate of UEs (Default) \\	
		\hline
		$N_{0}$ & Cellular thermal noise power spectral density \\	
		\hline
		$f_c$ & Cellular carrier frequency \\	
		\hline
		$f_c^\text{bk}$ & mmWave carrier frequency \\		
		\hline
		$B$ & Cellular carrier bandwidth  \\	
		\hline
		$B^\text{bk}$ & mmWave carrier bandwidth \\		
		\hline
		$\gamma_{\text{th}}$ & Cellular SINR threshold  \\ 	
		\hline
		$\gamma_{\text{th}}^\text{bk}$ & mmWave SINR threshold \\	
		\hline	
		$h_{\min}$ & Minimum altitude of each DBS \\	
		\hline
		$h_{\max}$ & Maximum altitude of each DBS  \\	
		\hline
	\end{tabular}
\end{table}

\subsubsection{GBS-to-DBS Propagation Model}

In the considered system model, the GBS uses directional mmWave antennas to transmit signals to the DBSs. Since the DBS is flying at a relatively high altitude, for the sake of simplification, we assume that backhaul GBS-to-DBS links experience the LoS propagation condition. The average path loss in dB of 28 GHz mmWave signal is given as~\cite{6834753}
\begin{equation}\label{eq:g2u:los}
L_{j,G}^{\text{bk}}=61.4 + 20\log_{10}\left(d_{j,G}\right),
\end{equation} 
where $d_{j,G}$ is the distance between the GBS and DBS $U_j$ in meters. Let $P_G^{\text{bk}}$ is the fixed transmit power of the mmWave antenna and $B_{j,G}^\text{bk}$ is the allocated bandwidth (in Hz) of the mmWave back-haul link for DBS $U_j$, according to~\eqref{eq:g2u:los}, the received signal-to-noise ratio (SNR) at a DBS is  
\begin{equation}\label{eq:g2u:sinr}
\gamma_{j,G}=\dfrac{P_G^{\text{bk}} \cdot 10^{-L_{j,G}^{\text{bk}}/10}}{B_{j,G}^\text{bk}N_0}\geq\gamma_{\text{th}}^\text{bk},
\end{equation} 
where $N_0$ is the thermal noise power spectral density and $\gamma_{\text{th}}^\text{bk}$ is a given threshold of mmWave back-haul transmissions.
According to the Shannon theorem, the back-haul capacity of a DBS $U_j$ can be obtained by
\begin{equation}\label{eq:backhaul_rate:uav}
\hat{C}_j=B_{j,G}^\text{bk}\log_2\left(1+\gamma_{j,G}\right).
\end{equation}
\vspace{-10pt}

\begin{figure}[!t]
	\centering
	\includegraphics[width=\columnwidth]{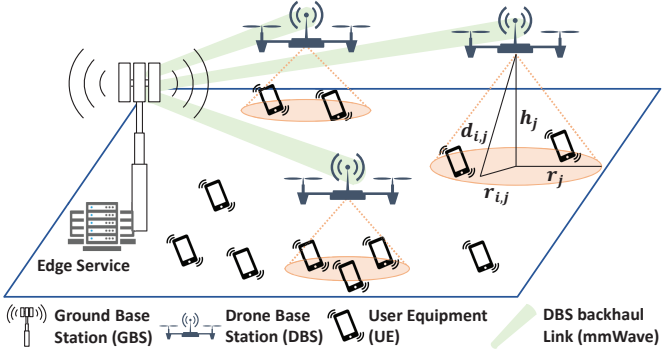}
	\caption{Architecture of the considered drone-assisted cellular system.}
	\label{fig:system}
\end{figure}

\subsubsection{DBS-to-UE Propagation Model}
The second propagation model is used to model the downlink transmission from a DBS to a UE.
Such a radio propagation model is well-known as the air-to-ground propagation channel and commonly modeled by considering the LoS and NLoS signals along with their occurrence probabilities separately~\cite{7417609}. We adopt the air-to-ground channel model in~\cite{6863654}, and the probabilities of LoS and NLoS for a UE $u_i$ associated with DBS $U_j$ are
\begin{align}\label{PNLoS_to_user}
P_{h_j,r_{i,j}}^\text{LoS}&=\dfrac{1}{1+a\exp\left(-b\left(\dfrac{180}{\pi}\tan^{-1}\left(\dfrac{h_j}{r_{i,j}}\right)-a\right)\right)},\nonumber\\
P_{h_j,r_{i,j}}^\text{NLoS}&=1-P_{h_j,r_{i,j}}^\text{LoS},
\end{align}
where $h_j$ is the altitude of each DBS $U_j$, $a$ and $b$ are environment variables, $r_{i,j}$ is the horizontal euclidean distance between $u_i$ and $U_j$. Note that $r_{i,j}=\sqrt{(x_j-x_i)^2+(y_j-y_i)^2}$, $(x_j,y_j)$ is the horizontal location of DBS $U_j$, $(x_{i},y_{i})$ is the horizontal location of UE $u_i$, $i=1,2,\dots,N$, and $j=1,2,\dots,k$. Considering the free space propagation loss, the channel model \cite{6863654} of the LoS and NLoS links can be written as
\begin{align}
L_{h_j,r_{i,j}}^\text{LoS}&=20\log_{10}\left(\dfrac{4\pi f_c d_{i,j}}{c}\right)+\eta_{LoS}, \nonumber\\
L_{h_j,r_{i,j}}^\text{NLoS}&=20\log_{10}\left(\dfrac{4\pi f_c d_{i,j}}{c}\right)+\eta_{NLoS},\label{model_nlos}
\end{align}
where $\eta_{LoS}$ (dB) and $\eta_{NLoS}$ (dB) are the mean additional losses for LoS and NLoS, $f_c$ is the carrier frequency of front-haul link, and $d_{i,j}=\sqrt{r_{i,j}^2+h_j^2}$ is the euclidean distance between $u_i$ and $U_j$. According to~\eqref{PNLoS_to_user} and~\eqref{model_nlos}, and let $\theta_{i,j}={\tan^{-1}}\left(h_j/r_{i,j}\right)$, we can obtain the average ATG channel model between $u_i$ and $U_j$ and it is denoted as
\begin{align}\label{average_atg_model}
L_{h_j,r_{i,j}}&=P_{h_j,r_{i,j}}^\text{LoS}L_{h_j,r_{i,j}}^\text{LoS}+P_{h_j,r_{i,j}}^\text{NLoS}L_{h_j,r_{i,j}}^\text{NLoS} \nonumber\\
&=\dfrac{\eta_{LoS}-\eta_{NLoS}}{1+a\exp\left(-b\left(\dfrac{180}{\pi}\theta_{i,j}-a\right)\right)} \nonumber\\
&+20\log_{10}\left({r_{i,j} \sec\theta_{i,j}}\right)+20\log_{10}\left(\dfrac{4\pi f_c}{c}\right)+\eta_{NLoS}.
\end{align}

Let $P_{i,j}$ be the minimum required transmit power for transmitting signal from DBS $U_j$ to UE $u_i$, the transmission is successful if the received signal-to-interference-plus-noise ratio (SINR) at a UE is larger than a certain threshold $\gamma_{\text{th}}$. Thus, SINR expression for UE $u_i$ associated with DBS $U_j$ is 
\begin{equation}\label{eq:sinr:d2u}
\gamma_{i,j}=\frac{P_{i,j}\cdot 10^{-L_{h_j,r_{i,j}}/10}}{I_G+I_{\mathcal{U}\setminus\{U_j\}}+B_{i,j}N_0}\geq\gamma_{\text{th}},
\end{equation}
where $I_G$ is the deterministic received interference power from the GBS and $I_{\mathcal{U}\setminus\{U_j\}}=\sum_{j'=1}^{k}P_{i,j'}\cdot 10^{-L_{h_j,r_{i,j'}}/10}\psi_{j,j'}$ is the interference power from the nearby DBSs if UE $u_i$ locates in the overlapping coverage, where $\psi_{j,j'}=1$ if $u_i$ locates in the overlapping coverage area of DBSs $U_j$ and $U_{j'}$, and $U_{j'}\in\mathcal{U}, \forall j'\neq j$; otherwise, $\psi_{j,j'}=0$.
According to the Shannon theorem and~\eqref{eq:sinr:d2u}, the allocated data rate (in bps) of $u_i$ associated with $U_j$ will be
\begin{equation}\label{eq:data_rate:uav2ue}
c_{i,j}=B_{i,j}\log_2\left(1+\gamma_{i,j}\right),
\end{equation}
where $B_{i,j}$ is the allocated bandwidth (in Hz) of down-link connection from DBS $U_j$ to a served UE $u_i$. The transmit power allocated to $u_i$ of interest can be obtained by
\begin{align}\label{eq:transmit_power:uav_to_1_ue}
P_{i,j}=&10^{L_{h_j,r_{i,j}}/10}\left(I_G+I_{\mathcal{U}\setminus\{U_j\}}+B_{i,j}N_0\right)\nonumber\\
&\times\left(2^{c_{i,j}/B_{i,j}}-1\right).
\end{align}
Then, the potential total transmit power of DBS $U_j$ for serving its associated UEs can be calculated as
\begin{equation}\label{eq:transmit_power:uav_to_all_ue}
P_j=\sum_{i=1}^{N_j}P_{i,j},
\end{equation}
where $N_j$ is the number of UEs associated with DBS $U_j$. According to~\eqref{eq:data_rate:uav2ue}, the data transmission rate of DBS $U_j$ for serving its associated UEs is
\begin{equation}\label{eq:sum_rate_constrait:uav_to_all_ue}
C_j=\sum_{i=1}^{N_j}c_{i,j}.
\end{equation}

\subsubsection{GBS-to-UE Propagation Model}
For the terrestrial wireless channel between points $p_1$ and $p_2$, we consider a standard power law path-loss $L_{p_1,p_2}=||p_1-p_2||^{-\alpha}$ with path-loss exponent $\alpha>2$. All the terrestrial propagation signals are assumed to experience independent Rayleigh fading. The GBS are assumed to transmit at fixed power $P_G$ for terrestrial communications. The received power of UE $u_i$ served by the GBS is therefore $P_Ghr_{i,G}^{-\alpha}$, where $h\sim\exp(1)$ models Rayleigh fading and $r_{i,G}$ is the horizontal distance between a UE and the GBS. Since there are $k$ DBSs in the considered system, the co-channel interference power experienced
by a UE can be expressed as
\begin{equation}\label{eq:interference:uav2ue}
I_\mathcal{U}=\sum_{j=1}^{k}P_{j}\cdot 10^{-L_{h_j,r_{i,j}}/10},
\end{equation}
where $P_j$ is the transmit power of DBS $U_j$ and $r_{i,j}$ is the distance from UE $u_i$ to DBS $U_j$.
The SINR expression for a user $u_i$ that can connect to the GBS is
\begin{equation}\label{eq:sinr:g2u}
\gamma_{i,G}=\frac{P_Ghr_{i,G}^{-\alpha}}{I_\mathcal{U}+B_{i,G}N_0}\geq\gamma_{\text{th}}.
\end{equation}
The achievable data rate (in bps) of a UE associated with the GBS can be calculated as
\begin{equation}\label{eq:rate:gbs_ue}
c_{i,G}=B_{i,G}\log_2\left(1+\gamma_{i,G}\right),
\end{equation}
where $B_{i,G}$ is the allocated bandwidth (in Hz) to $u_i$ associate with the GBS. The potential transmission rate (in bps) of the GBS can be obtained by~\cite{6042301}
\begin{equation}\label{eq:rate:gbs}
C_G=\frac{\lambda_G}{\pi r_G^2}\overline{c_{i,G}}=\sum_{i=1}^{N_G}c_{i,G},
\end{equation}
where $r_G$ is the coverage radius of the GBS, $\lambda_G$ is the UE density of GBS's service coverage, $\overline{c_{i,G}}$ is the average data rate of a UE associated with the GBS, and $N_G$ is the number of UEs which is associated with the GBS.

\subsection{Problem Formulation}
\label{sec:problem:formulation}
In this work, the considered system model is depicted in Fig.~\ref{fig:system}. We focus on the case of deploying multiple DBSs in the target area to improve the downlink sum rate of the drone-assisted cellular system with one GBS. The placement of DBSs must satisfy the minimum data rate requirement of UE which predefined by the cellular operator.
Thus, we consider the optimization problem from the cellular operator's (or service provider's) perspective. Due to the limited budget, the cellular operator always tries to use the minimum number of DBSs to improve the overall system sum rate and meets the minimum data rate requirement of UE. We refer such a problem as \emph{Minimizing the Number of Required DBSs} (MNRD) problem and it can be defined as follows.

\begin{definition}[MNRD Problem]\label{def:p0:num_of_dbs}
	Suppose that the notations and assumptions are defined as above. Given a large number of arbitrary distributed UEs $N$ in the target area, the optimization problem of minimizing the number of required DBSs will be
	
	\noindent		
	\begin{align}\label{eq:min_num_required_dbs}
	\min&\enspace k \tag{P1}\\
	s.t.&\enspace \max_{N_G,N_j}\left(\sum_{l=1}^{N_G}c_{l,G}+\sum_{j=1}^{k}\sum_{i=1}^{N_j}c_{i,j}\right)\geq \tau N c_{\min}, \label{eq:min_num_required_dbs:c1}\\
	&\enspace \tau N \leq N_G+\sum_{j=1}^k N_j\leq N. \label{eq:min_num_required_dbs:c2}
	\end{align}
	where $\tau\in[0,1]$.
\end{definition}
	
The constraint~\eqref{eq:min_num_required_dbs:c1} is a sub-function to find the placement of $k$ DBS to maximize the system sum rate and the obtained placement result must satisfy the total demand of data rate, where $c_{\min}$ is the demand of UE for the minimum data rate from the perspective of cellular operators and $\tau$ is a predefined target satisfaction rate (in percentage).
Constraint~\eqref{eq:min_num_required_dbs:c2} stipulates that the percentage of satisfied UEs should be more than a predefined threshold $\tau$ and each UE only can associate with one DBS or the GBS at a time.
We then refer sub-function~\eqref{eq:min_num_required_dbs:c1} as \emph{System Sum rate Optimization} (SSO) problem~\eqref{eq:system_sum_rate} and its detailed definition is presented as follows.
\begin{definition}[SSO Problem]\label{def:p1:system_sum_rate}
	Suppose that the notations and assumptions are defined as above, the SSO problem is to use a given number of DBSs $k$ for finding the appropriate $N_j$ and $N_G$ such that
	\begin{align}\label{eq:system_sum_rate}
	\max_{N_G,N_j}&\enspace\left(\sum_{l=1}^{N_G}c_{l,G}+\sum_{j=1}^{k}\sum_{i=1}^{N_j}c_{i,j}\right),&\tag{P2}\\
	s.t.&\enspace h_{\min}\leq h_j\leq h_{\max},&\label{eq:system_sum_rate:c1}\\
	&\enspace r_{\min}\leq r_j\leq r_{\max},&\label{eq:system_sum_rate:c2}\\
	&\enspace c_{l,G}\geq c_{\min}, \,\qquad l=1,2,\dots,N_G, \label{eq:system_sum_rate:c3}\\
	&\enspace \sum_{l=1}^{N_G}c_{l,G}\leq \hat{C}_G, &\label{eq:system_sum_rate:c4}\\
	&\enspace c_{i,j}\geq c_{\min}, \qquad\; i=1,2,\dots,N_j, j=1,2,\dots,k, \label{eq:system_sum_rate:c5}\\
	&\enspace \sum_{i=1}^{N_j}c_{i,j}\leq \hat{C}_j, \quad j=1,2,\dots,k.\label{eq:system_sum_rate:c6}
	\end{align}
\end{definition}

In constraint~\eqref{eq:system_sum_rate:c1}, the deployed altitude $h_j$ of each DBS is only allowed within $[h_{\min},h_{\max}]$ which depends on the limitations of local laws and ability of drone. 
In the considered system, the allowable path-loss of each UE is set to a fixed value, $L_\text{allowable}$. 
If $h_j$ is given, the corresponding coverage of $U_j$, $r_j$, in constraint~\eqref{eq:system_sum_rate:c2}, can be determined by solving following equations~\cite{6863654}:
\begin{align}
\footnotesize
\dfrac{\pi}{9\ln\left(10\right)}\tan \theta_j+&\frac{ab\Big(\eta_{LoS}-\eta_{NoS}\Big)\exp{\left(-b\left(\dfrac{180}{\pi}\theta_j-a\right)\right)}}{\left(a\exp{\left(-b\left(\dfrac{180}{\pi}\theta_j-a\right)\right)}+1\right)^2}=0, \label{eq:optimal_altitude_1}\\
\theta_j=&\tan^{-1}\left(h_j/r_j\right). \label{eq:optimal_altitude_2}
\end{align}
Note that the constraints $r_{\min}$ and $r_{\max}$ can be determined in the same way by using the predefined constraints $h_{\min}$ and $h_{\max}$ as the input.
In constraints~\eqref{eq:system_sum_rate:c3} and~\eqref{eq:system_sum_rate:c5}, the GBS $G$ or each DBS $U_j$ needs to guarantee the minimum allocated date rate of a UE, $c_{\min}$.
Constraint~\eqref{eq:system_sum_rate:c4} guarantees that the total downlink transmission rate of the links from the GBS and its associated UEs does not exceed the maximum ability of providing data rate $\hat{C}_G$. Constraint~\eqref{eq:system_sum_rate:c6} is used to make the total downlink transmission rate of the links from DBS $U_j$ to its associated UEs do not exceed the maximum allocated data rate of back-haul link on $U_j$ according to~\eqref{eq:backhaul_rate:uav}. 
The interference is mitigable but unavoidable in such a heterogeneous network consisting of dense flash crowds and DBSs, so it is too hard to guarantee that all the UEs can always have the satisfied minimum data rate under different dense flash crowd scenarios.  

In fact, the considered original problems~\eqref{eq:min_num_required_dbs} and~\eqref{eq:system_sum_rate} are formulated from the resource allocation perspective. In the conventional terrestrial cellular system, the locations of base stations are fixed so that most existing solutions focus on enhancing mechanisms of UE association, resource scheduling, or power control to improve the system performance.
However, for the drone-assisted cellular system, the system can maintain the UE association by utilizing an additional dimension, i.e. the altitude of DBSs. 
Hence, the considered problem can be solved differently. Such a solution is more flexible and it can utilize the 3D space for deploying multiple DBSs to satisfy the dynamic UE demands under different environmental conditions.

From the perspective of the DBS placement problem, we are going to simplify and reformulate the original problems~\eqref{eq:min_num_required_dbs} and~\eqref{eq:system_sum_rate} below. We do not consider the power control issue, so the transmit power of each DBS is set to a fixed value. 
Let $\rho_{i,j}$ and $\rho_{i,G}$ be the indicator functions to identify a UE is associated with the DBS $U_j$ or GBS $G$, defined as
\begin{align}
\rho_{i,j}=&\begin{cases}
1, & \text{if }r_{i,j}\leq r_j\text{ and }\gamma_{i,j}\geq\gamma_{\text{th}}.\\
0, & \text{otherwise},
\end{cases}\label{eq:indicator_ij}\\
\rho_{i,G}=&1-\sum_{j=1}^{k}\rho_{i,j}.\label{eq:indicator_ig}
\end{align}
With the indicator functions $\rho_{i,j}$ and $\rho_{i,G}$, the number of associations from UE to the GBS or DBS $U_j$ are $N_G=\sum_{i=1}^{N}\rho_{i,G}$ and $N_j=\sum_{i=1}^{N}\rho_{i,j}$, respectively. Then, the original problems~\eqref{eq:min_num_required_dbs} and~\eqref{eq:system_sum_rate} can be reformulated as following \emph{Decision of the Number of Required DBSs} (DNRD) problem and \emph{Drone Placement decision with the Maximum Sum Rate} (DPMSR) problem, respectively.

\begin{definition}[DNRD Problem]\label{def:p3:num_of_dbs}
	Suppose that the notations and assumptions are defined as above. The minimum number of required DBSs will be 
	\begin{align}\label{eq:decition_num_required_dbs}
	\min&\enspace k \tag{P3}\\
	s.t.&\enspace \max_{x_j,y_j,r_j,h_j,\forall j} \left(\sum_{l=1}^{N}c_{l,G}\rho_{i,G}+\sum_{j=1}^{k}\sum_{i=1}^{N}c_{i,j}\rho_{i,j}\right)\geq \tau N c_{\min}, \label{def:p3:c1}\\
	&\enspace \tau N\leq\sum_{i=1}^{N}\rho_{i,G}+\sum_{j=1}^{k}\sum_{i=1}^{N}\rho_{i,j}\leq N. \label{def:p3:c2}
	\end{align}
	where $\tau\in[0,1]$.
\end{definition}

\begin{definition}[DPMSR Problem]\label{p:dpmsr}
	Suppose that the notations and assumptions are defined as above, the DPMSR problem is search for the appropriate placement parameters $(x_j,y_j,h_j)$ of each DBS, $\forall j=1,2,\dots,k$, such that	
	
	\noindent		
	\begin{align}\label{eq:system_sum_rate:decision_ver}
	\max_{x_j,y_j,r_j,h_j,\forall j}&\enspace\left(\sum_{i=1}^{N}c_{i,G}\rho_{i,G}+\sum_{j=1}^{k}\sum_{i=1}^{N}c_{i,j}\rho_{i,j}\right),&\tag{P4}\\
	s.t.&~\eqref{eq:system_sum_rate:c1},~\eqref{eq:system_sum_rate:c2},~\eqref{eq:indicator_ij},~\eqref{eq:indicator_ig} \nonumber\\
	&\enspace c_{i,j}\rho_{i,j}\geq c_{\min}\rho_{i,j},\quad\quad\quad\;\; i=1,2,\dots,N,\nonumber\\
	&\hspace{10.65em}\quad j=1,2,\dots,k,&\label{eq:system_sum_rate:decision_ver:c3}\\	
	&\enspace \sum_{i=1}^{N}c_{i,j}\rho_{i,j}\leq \hat{C}_j, \quad\quad\quad\quad j=1,2,\dots,k,\label{eq:system_sum_rate:decision_ver:c4}\\
	&\enspace c_{i,G}\rho_{i,G}\geq c_{\min}\rho_{i,G},\quad\quad\;\;\; i=1,2,\dots,N,\label{eq:system_sum_rate:decision_ver:c5}\\
	&\enspace \sum_{i=1}^{N}c_{i,G}\rho_{i,G}\leq \hat{C}_G. \label{eq:system_sum_rate:decision_ver:c6}
	\end{align}
\end{definition}


\subsection{Feasibility}
The considered DPMSR problem in~\eqref{eq:system_sum_rate:decision_ver} is feasible while the minimum data rate requirement $c_{\min}$ and the satisfaction rate $\tau$ are set to zero. For instance, consider a general instance of the problem. Take a random position for each DBS $U_j$ inside the target area and both altitude and coverage of $U_j$ do not violate constraints~\eqref{eq:system_sum_rate:c1} and~\eqref{eq:system_sum_rate:c2}. Now if we set all the binary variables $\{\rho_{i,j}\}$ equal to $0$, all the binary variables $\{\rho_{i,G}\}$ will be equal to $1$ according to~\eqref{eq:indicator_ig}. Since the SINR functions $\gamma_{i,j}$ and $\gamma_{i,G}$ are positive (see~\eqref{eq:sinr:d2u} and~\eqref{eq:sinr:g2u}, respectively), the solution satisfies all the constraints. Such a placement result means that all the UEs are served by the GBS.
If $c_{\min}$ and $\tau$ are not equal $0$, problem~\eqref{eq:system_sum_rate:decision_ver} may not converge for a given number of DBSs $k$. Note that~\eqref{eq:system_sum_rate:decision_ver} is a sub-function/sub-problem~\eqref{def:p3:c1} of~\eqref{eq:decition_num_required_dbs}. When~\eqref{eq:system_sum_rate:decision_ver} does not converge, it means that $k$ is too small and the system will iteratively use $k=k+1$ to execute~\eqref{eq:system_sum_rate:decision_ver} until~\eqref{eq:system_sum_rate:decision_ver} converges. If~\eqref{eq:system_sum_rate:decision_ver} still cannot converge when $k=k_{\max}$, where $k_{\max}$ is the upper bound of $k$ we propose in Section~\ref{sec:initialization}. It means that the cellular operator cannot provide a feasible placement for satisfying the given constraints $c_{\min}$ and $\tau$. In this case, the cellular operator needs to relax the constraints $c_{\min}$ and $\tau$ to search a feasible placement.

\vspace{-10pt}
\subsection{NP-Hardness}
In this subsection, we then prove that the considered SSO problem~\eqref{eq:system_sum_rate} is an NP-hard problem. 
If the DPMSR problem~\eqref{eq:system_sum_rate:decision_ver} is NP-complete, it implies that the SSO problem is NP-hard. 
To show further that the DPMSR problem is NP-complete, we consider the special case of it with the relaxation on some constraints. That is, we consider the case $N_G=\sum_{i=1}^{N}\rho_{i,G}=0$ that all the UEs only served by the deployed $k$ DBSs. We also relax the constraints on spatial limitations and the required minimum data rates.
The following gives the definition of the decision problem for the above special case.
\begin{definition}[DPMSR$\rho$ problem]\label{p:dpmsr:p2.1}
	\textbf{Instance:} Suppose that the notations and assumptions are defined as above. All the UEs are served by the deployed $k$ DBSs and $k<N$. Without the consideration of the constraints on spatial limitations and the required minimum data rates, the DPMSR$\rho$ problem can be formally defined as 	\begin{align}\label{p:dpmsr:p2.1:eq}
	\max_{\rho_{i,j}, \forall i,j}&\enspace\sum_{j=1}^{k}\sum_{i=1}^{N}c_{i,j}\rho_{i,j},&\nonumber\\
	s.t.&\enspace \sum_{i=1}^{N}c_{i,j}\rho_{i,j}\leq \hat{C}_j,\quad j=1,2,\dots,k,\nonumber\\
	&\enspace \sum_{j=1}^{k}\rho_{i,j}\leq 1, \qquad\quad i=1,2,\dots,N,\nonumber
	\end{align}
	where 
	\[
	\rho_{i,j}=\begin{cases}
	1, & \text{if UE $u_i$ is associated with DBS $U_j$;}\\
	0, & \text{otherwise}.
	\end{cases}
	\]
\end{definition}

To show that the DPMSR$\rho$ problem is NP-complete, we reduce the \emph{0--1 Multiple Knapsack Problem} (MKP) to the DPMSR$\rho$ problem. The MKP problem is defined as follows.
\begin{definition}[0--1 MKP]\label{p:mkp}
	\textbf{Instance:} Given a set of $n$ items and a set of
	$m$ knapsacks ($m<n$), let $p_{i,j}$ be the profit of item $i$, $w_{i,j}$ be the weight of item $i$, $v_j$ be the volume of knapsack $j$, and then select $m$ disjoint subsets of items so that the total profit of selected items is a maximum and each subset can be assigned to a different knapsack whose volume is no less than the total weight of items in the subset. Formally,
	\begin{align}\label{p:mkp:eq}
	\max_{\varphi_{i,j}, \forall i,j}&\enspace\sum_{j=1}^{m}\sum_{i=1}^{n}p_{i,j}\varphi_{i,j},&\nonumber\\
	s.t.&\enspace \sum_{i=1}^{n}w_{i,j}\varphi_{i,j}\leq v_j, \quad j=1,2,\dots,m,\nonumber\\
	&\enspace \sum_{j=1}^{m}\varphi_{i,j}\leq 1, \qquad\quad i=1,2,\dots,n,\nonumber
	\end{align}
	where 
	\[
	\varphi_{i,j}=\begin{cases}
	1, & \text{if item $i$ is assigned to knapsack $j$;}\\
	0, & \text{otherwise}.
	\end{cases}
	\]
\end{definition}

\begin{thm}
The DPMSR$\rho$ problem is NP-complete.
\end{thm}
\begin{proof}
	It is easy to see that the DPMSR$\rho$ problem is in NP, since validating the existence of a given placement simply needs polynomial time. In order to prove the DPMSR$\rho$ problem is NP-complete, a reduction from MKP can be made. Suppose that $I'$ is an instance of the MKP A corresponding instance $I$ of the DPMSR$\rho$ problem can be constructed from $I'$ as follows.
	\begin{enumerate}
		\item Let item $i$ correspond to UE $u_i$ and the number of items be the number of UEs ($n=N$).
		\item Let knapsack $j$ correspond to DBS $U_j$ and the number of knapsacks be the number of DBSs ($m=k$).
		\item Let the profit of item $p_{i,j}$ be the allocated data rate of UE $c_{i,j}$ and consider the case of a item's profit is identical to its weight $w_{i,j}$ ($c_{i,j}=p_{i,j}=w_{i,j}$).
		\item Let the selection function $\varphi_{i,j}$ correspond to the indicator function $\rho_{i,j}$.
	\end{enumerate}
	It is straightforward to show that there is a solution for an instance $I$' of the MKP if and only if there is a solution for instance $I$ of the DPMSR$\rho$ problem since the reduction is a one-to-one mapping for the variables from the MKP to the DPMSR$\rho$ problem. Hence, the DPMSR$\rho$ problem is NP-complete.
\end{proof}

Thus, we can conclude the following theorem.
\begin{thm}
The DPMSR problem is NP-complete and it implies that the SSO problem is NP-hard.
\end{thm}

\section{The Data-Driven 3D Placement (DDP)}
\label{sec:3Dplacement}
The considered problem~\eqref{eq:system_sum_rate:decision_ver} presents a non-convex formulation since the data rate is related to the quality of received signals and the attenuation of signals depends on the LoS probability. Currently, no off-the-shelf optimizer can be applied to solve this optimization problem. We also have proved that the considered problem is a NP-hard problem. Even though we ignore the LoS probability, the time complexity of solving this problem with an exhaustive search is generally of exponential order. However, due to the limitation of time, the optimal solution is not suitable for the considered case of serving outdoor temporary events with flash crowds. The placement decision must be very quick so that the obtained placement result can be useful for target UEs. Therefore, in this work, we propose a model-free algorithm, data-driven 3D placement (DDP), to solve the considered DBS placement problem. The algorithm can be executed in polynomial time. DDP can effectively improve the sum rate performance of the drone-assisted cellular system in a more efficient way, especially for the unpredictable events or flash crowds with arbitrary distributed users.

The proposed approach uses a three-stage procedure with the input spatial information of UEs and a GBS to provide an effective placement of multiple DBSs. 
The pseudo-code of the proposed placement procedure is described as Algorithm~\ref{alg:ddp}. In addition, we explain the notations/variables in an in-text manner and use some comment texts to help the ease of understanding. The detailed descriptions of the proposed approach will be presented in following subsections.
%

\subsection{Initialization}
\label{sec:initialization}
According to the considered DPMSR problem~\eqref{eq:system_sum_rate:decision_ver}, we can know that the system sum rate mainly depends on $N_j=\sum_{i=1}^{N}\rho_{i,j}$ and $N_G=1-\sum_{j=1}^{k}N_j$ which are determined by the placement of DBSs. It is also similar to user association or load balancing issues of communication systems. The proposed approach uses the spatial information of UEs, DBSs, and the GBS to provide an effective placement of DBSs. 
Let variable $L_G=(x_G,y_G)$ record the location (coordinate) of the GBS, a set $L_E$ store the locations of UEs, and a set $L_U$ save the locations of DBSs. In the system initialization stage, the system computes and store the received power of each UE from the GBS, $P_{i,G}^R=P_Ghr_{i,G}^{-\alpha}$ in a set $S_G$, where $1\leq i\leq N$. The distances from the GBS to all UEs are stored in a set $D_G$. 

In the initialization stage of the proposed approach, the system first determines the preliminary association between the GBS and each UE since DBSs are used to assist the GBS. In fact, the interference power also cannot be obtained at the step. The initial association between the GBS and each UE then will be determined by the condition $\gamma_{i,G}\geq\gamma_{\text{th}}$, where $\gamma_{i,G}$ is the SINR without considering the interference. We can use the indicator function~\eqref{eq:indicator_ig} to identify a UE is served by the GBS and it is equivalent to 
\begin{align}
\rho_{i,G}=&\begin{cases}
1, & \text{if }r_{i,G}\leq r_G\text{ and }\gamma_{i,G}\geq\gamma_{\text{th}},\\
0, & \text{otherwise}.
\end{cases}\label{eq:indicator_iG}
\end{align}
The number of UEs having at least SINR value, $\gamma_{\text{th}}$, is denoted as
\vspace{-5pt}
\begin{equation}\label{eq:temporary_n_g_max}
N_G^{\text{temp}}=\sum_{i=1}^{N}\rho_{i,G}.
\end{equation}
In addition, according to~\eqref{eq:system_sum_rate:c3}, and~\eqref{eq:system_sum_rate:c4}, we can obtain the upper bound of $N_G$ by
\begin{equation}\label{eq:initial_n_g_max}
N_G^{\max}=\hat{C}_G/c_{\min}.
\end{equation}
In the proposed approach, we select the UEs with top $N_G$ values of the SINR to associate with the GBS, where $N_G$ is 
\begin{equation}\label{eq:initial_n_g}
N_G=\min\left(N_G^{\text{temp}},N_G^{\max}\right).
\end{equation}

After initializing the association between the GBS and UEs, the system will place a number of DBSs to serve the remaining dissociated UEs. However, the number and locations of deployed DBSs are unknown. We have formulated such a problem as~\eqref{eq:min_num_required_dbs} which can be converted into~\eqref{eq:decition_num_required_dbs}. If the number of DBSs $k$ is unknown and not given, the system needs to interactively search for the appropriate value of $k$ starting from $k=1$. Such an exhaustive search wastes a large amount of computational cost, especially when the scale of the considered scenario is large. To reduce the computational cost and accelerate the converge speed, we propose following two propositions.
\begin{proposition}[Lower bound of $k$]\label{proposition1}
	Suppose that the notations are defined as above, 	given a predefined minimum data requirement of UE, $c_{\min}$, and a target satisfaction rate, $\tau$, the lower bound of $k$ will be 
	\begin{equation*}
	k_{\min}=\left\lceil\dfrac{\tau(N-N_G) c_{\min}}{B\log_2(1+\gamma_{\text{th}})}\right\rceil,
	\end{equation*}
	where $B$ is the bandwidth of front-haul channel provided by a DBS.
\end{proposition}
\begin{IEEEproof}
	In the considered system, the number of UEs, $N$, and the minimum data rate requirement, $c_{\min}$, are given beforehand. After the initialization stage, the number of UEs served by the GBS, $N_G$, is obtained by~\eqref{eq:initial_n_g}. With the above information, we can know that the expected number of remaining UEs needs to be served by DBSs is $\tau(N-N_G)$ and the expected minimum total traffic demand is $\tau(N-N_G)c_{\min}$. If we deploy $k$ DBSs to serve these remaining UEs, the corresponding number of UEs served by DBS $U_j$ will be $N_j\in \{N_1,N_2,\dots,N_k\}, j=1,2,\dots,k$. Let $N_j^{\max}=\max\{N_1,N_2,\dots,N_k\}$, the minimum (lower bound) of $k$ will be
	\begin{equation}\label{eq:proposition1:k_min1}
	k_{\min} = \left\lceil\dfrac{\tau(N-N_G)}{N_j^{\max}}\right\rceil,
	\end{equation}
	where $(N-N_G)$ is an integer value obtained after associating UEs with the GBS. In addition, according to the given $c_{\min}$ and~\eqref{eq:sum_rate_constrait:uav_to_all_ue},~\eqref{eq:data_rate:uav2ue} can be rewritten as
	\begin{equation}\label{eq:proposition1:c_min}
		c_{\min} = \dfrac{B}{N_j^{\max}}\log_2(1+\gamma_{\text{th}}),
	\end{equation}
	where $B$ is the channel bandwidth provided by a DBS.
	In summary, according to~\eqref{eq:proposition1:c_min}, the lower bound of $k$ in~\eqref{eq:proposition1:k_min1} can be rewritten as
	\begin{equation}\label{eq:initial_value_of_k}
	k_{\min} = \left\lceil\dfrac{\tau(N-N_G) c_{\min}}{B\log_2(1+\gamma_{\text{th}})}\right\rceil.
	\end{equation}
	In the considered system, $N$, $c_{\min}$, $B$, and $\gamma_{\text{th}}$ are given beforehand and $N_G$ is determined after the initial stage, the lower bound of $k$ can be derived by~\eqref{eq:initial_value_of_k}. 
\end{IEEEproof}

However, each DBS has a capacity constraint~\eqref{eq:system_sum_rate:c6} and the spectrum of backhaul links is limited. Hence, we can give the following proposition to find the upper bound of $k$.
\begin{proposition}[Upper bound of $k$]\label{proposition2}
	Suppose that the notations are defined as above, 	given a predefined SINR threshold, $\gamma_{\text{th}}^{\text{bk}}$, the upper bound of $k$ will be 
	\begin{equation*}
	k_{\max}=\left\lfloor\dfrac{B^\text{bk}\log_2(1+\gamma_{\text{th}}^{\text{bk}})}{B\log_2(1+\gamma_{\text{th}})}\right\rfloor,
	\end{equation*}
	where $B^\text{bk}$ is the channel bandwidth of mmWave backhaul transmission shared by all DBSs.
\end{proposition}
\begin{IEEEproof}
	According to~\eqref{eq:sum_rate_constrait:uav_to_all_ue} and~\eqref{eq:proposition1:c_min}, we can rewrite~\eqref{eq:system_sum_rate:c6} as
	\begin{equation}\label{eq:proposition2:backhaul_constraint}
	C_j=\sum_{i=1}^{N_j}c_{i,j}=c_{\min}N_j^{\max}=B\log_2(1+\gamma_{\text{th}})\leq \hat{C}_j, 
	\end{equation}
	where $j=1,2,\dots,k$. By~\eqref{eq:backhaul_rate:uav},~\eqref{eq:proposition2:backhaul_constraint} can be formulated as
	\begin{equation}\label{eq:proposition2:relation_between_fronthaul_backhaul}
	B\log_2(1+\gamma_{\text{th}})\leq \dfrac{B^\text{bk}}{k}\log_2(1+\gamma_{\text{th}}^{\text{bk}}).
	\end{equation}	
	After moving the items in~\eqref{eq:proposition2:relation_between_fronthaul_backhaul}, we can get
	\begin{equation*}
	k\leq \dfrac{B^\text{bk}\log_2(1+\gamma_{\text{th}}^{\text{bk}})}{B\log_2(1+\gamma_{\text{th}})}.
	\end{equation*}
	Thus, the upper bound of $k$ can be expressed as
	\begin{equation}\label{eq:proposition2:upperbound_k}
	k_{\max}=\left\lfloor\dfrac{B^\text{bk}\log_2(1+\gamma_{\text{th}}^{\text{bk}})}{B\log_2(1+\gamma_{\text{th}})}\right\rfloor.
	\end{equation}
\end{IEEEproof}

In summary, according to our two propositions, if the target satisfaction rate, $\tau$, and the minimum data rate requirement, $c_{\min}$, given by cellular operator does not make the following equation hold, it means that the considered system can not provide a feasible placement and the cellular operator needs to relax the constraints $\tau$ and $c_{\min}$.
\begin{equation*}
\left\lceil\dfrac{\tau(N-N_G) c_{\min}}{B\log_2(1+\gamma_{\text{th}})}\right\rceil\leq k \leq \left\lfloor\dfrac{B^\text{bk}\log_2(1+\gamma_{\text{th}}^{\text{bk}})}{B\log_2(1+\gamma_{\text{th}})}\right\rfloor.
\end{equation*}

\subsection{User Association Clustering}
In the second stage, the system makes each UE be associated with a least one DBS in a best-effort manner. For UE $u_i$, $\gamma_{i,j}$ must be not smaller than the given threshold $\gamma$ so that $u_i$ can be associated with a DBS $U_j$. In general, the allocated data rate $c_i$ and $\gamma_{i,j}$ of $u_i$ increases when distance $r_{i,j}$ between $u_i$ and $U_j$ decreases. We thus reduce the NP-complete 0-1 MKP to a variation of assignment problem, \emph{Capacitated Clustering Problem} (CCP)~\cite{MULVEY1984339}. The CCP is also NP-complete and can be defined as follows.
\begin{definition}[CCP]\label{p:ccp}
	\textbf{Instance:} Given a set of $N$ UEs and a set of $k$ DBS ($k<N$), let $r_{i,j}$ be the horizontal distance between UE $u_i$ and DBS $U_j$ (cluster centroid), $c_{i,j}$ be the allocated data rate of UE $u_i$, $\hat{C}_j$ be the back-haul constraint of DBS $U_j$, and then find $k$ disjoint subsets of UEs so that the total horizontal distance value of selected UEs is a minimum and each subset can be assigned to a different DBS whose back-haul constraint is no less than the total horizontal distance value of UEs in the subset. Formally,	
	\begin{align}\label{p:gap:eq}
	\min_{\rho_{i,j},\beta_j,\forall i,j}&\enspace\sum_{j=1}^{k}\sum_{i=1}^{N}r_{i,j}\rho_{i,j},&\\
	s.t.&\enspace \sum_{i=1}^{N}c_{i,j}\rho_{i,j}\leq \hat{C}_j, \quad j=1,2,\dots,k,\nonumber\\		
	&\enspace \sum_{j=1}^{k}\rho_{i,j}=1, \qquad\quad i=1,2,\dots,N,\nonumber\\
	&\enspace \sum_{j=1}^{k}\beta_j=k, \qquad\quad j=1,2,\dots, k,\nonumber
	\end{align}
	where $\beta_j\in\{0,1\}$ indicates whether DBS $U_j$ is deployed or not and 
	\[
	\rho_{i,j}=\begin{cases}
	1, & \text{if item $i$ is assigned to DBS $U_j$;}\\
	0, & \text{otherwise}.
	\end{cases}
	\]
\end{definition}

According to Definition~\ref{p:ccp} of CCP, we can know that the clustering technologies can be used to deal with the user association problem. Note that $r_{i,j}$ in~\eqref{p:gap:eq} represents the cost function for the clustering. We can substitute a customized cost function for $r_{i,j}$ to obtain a different clustering result. In this stage, we adopt a \emph{balanced $k$-means clustering}~\cite{10.1007/978-3-662-44415-3_4} for obtaining a balanced placement result. Such a result can make each DBS serve almost the same number of UEs. In fact, the data rate of a UE is not proportion to distance squared according to~\eqref{eq:sinr:d2u} and~\eqref{eq:data_rate:uav2ue} and the co-channel interference from different nearby DBSs is not considered in this stage. Balanced $k$-means clustering cannot find the optimal result for the considered problem~\eqref{eq:system_sum_rate:decision_ver}. However, the SINR value of each UE can only be derived after obtaining the candidate location and altitude of DBSs. We thus proposed a re-association stage for solving the above issue and we present the detail of the procedure in the following subsection.

\begin{algorithm2e}[!t]
	\footnotesize
	\SetAlgoLined
	\KwIn{dataset of UE locations $L_E$, location of the GBS $L_G$, the maximum number of DBS $K$, SINR threshold $\gamma_{\text{th}}$, the channel bandwidth provided by each DBS and the GBS $B$, the transmit power of the GBS $P_G$, the transmit power of a DBS $P_{\text{drone}}$, and the minimum data rate requirement $c_{\min}$}
	\KwOut{association information $L_{\text{association}}$, UEs' SINR $L_{\text{SINR}}$, DBS locations $L_{\text{drone}}^{\text{cand}}$, DBS coverage radii $L_{\text{radius}}$, and DBS altitudes $L_{\text{altitude}}$}
	create a list $L_{\text{drone}}^{\text{cand}}$ to store the DBS locations\;
	$N\leftarrow L_E.length$\;
	create a list $D_G$ to record the distance between the GBS and each UE\;
	create a list $S_G$ to save the received power from the GBS on each UE\;	
	create a list $L_{\text{association}}$ to save the association information of each UE\;	
	create a list $L_{\text{SINR}}$ to save UEs' SINR received from its associated DBS\;	
	create a list $L_{\text{radius}}$ to save the candidate coverage radius of each DBS\;	
	create a list $L_{\text{altitude}}$ to save the candidate altitude of each DBS\;
	\For{$i=1$ to $N$}{
		$D_G[i]\leftarrow\sqrt{(L_E[i].x-L_G.x)^2+(L_E[i].y-L_G.y)^2}$\;
		\tcc{$h$ modeles Rayleigh fading and the path-loss exponent $\alpha>2$}	
		$S_G[i]\leftarrow P_Gh*(D_G[i])^{-\alpha}$\;
		compute SINR $\gamma_{i,G}$ by~\eqref{eq:sinr:g2u} with the interference $I_U=0$ (mW)\;
		\If{$\gamma_{i,G}\geq\gamma_{\text{th}}$}{
			\tcc{The association value is $0,1,\dots,k$}
			$L_{\text{association}}[i]\leftarrow 0$\;
		}
	}
	get $N_G$ by checking the number of ``0" in $L_{\text{association}}$\; 
	initialize $k$ by~\eqref{eq:initial_value_of_k}\;
	run balanced $k$-means clustering~\cite{10.1007/978-3-662-44415-3_4} to cluster the dissociated UEs with $L_G$ and update $L_{\text{association}}$\label{alg:ddp:19}\;
	\Repeat{all the DBS locations in $L_{\text{drone}}^{\text{cand}}$ do not change\label{alg:ddp:20}}{
		do the placement refinement by finding the minimum enclosing circle of each cluster~\cite{10.1007/BFb0038202}\;
		update $L_{\text{drone}}^{\text{cand}}$ using the centor point of each minimum enclosing circle\;		
		update $L_{\text{radius}}$ using the radius of each minimum enclosing circle\;
		update $L_{\text{altitude}}$ by~\eqref{eq:optimal_altitude_1} and~\eqref{eq:optimal_altitude_2} using the corresponding radius in $L_{\text{radius}}$ as the input\;
		update the SINR value of each UE in $L_{\text{SINR}}$ by~\eqref{eq:sinr:d2u} and~\eqref{eq:sinr:g2u}\;
		\For{$i=1$ to $N$}{
			\If{$L_{\text{association}}[i]==-1\vee L_{\text{SINR}}[i]<\gamma_{\text{th}}$}{
				try to re-assocaite UE $u_i$ with another nearby DBS and update $L_{\text{association}}$ if the SINR value is not smaller than $\gamma_{\text{th}}$\;
				\uIf{exist another one DBS can be re-assocaited by $u_i$}{
					update $L_{\text{association}}[i]$ and	jump to line~\ref{alg:ddp:20}\;
				}\Else{
					$L_{\text{association}}[i]\leftarrow -1$\;
				}				
			}
		}
	}
	\If{\eqref{def:p3:c2} does not hold}{
		$k=k+1$\;
		jump to line~\ref{alg:ddp:19}\;
	}
	\Return $L_{\text{association}}$,$L_{\text{drone}}^{\text{cand}}$, $L_{\text{SINR}}$, $L_{\text{radius}}$, and $L_{\text{altitude}}$\;
	\caption{The procedure of DDP}
	\label{alg:ddp}
\end{algorithm2e}

\subsection{UE Re-association for Placement Refinement}
After the association stage, the system gets an initial placement recommendation with $k$ centroid points of the generated clusters. If we treat the horizontal coverage of each DBS mapping to the ground as an ideal circle and directly deploy each DBS to the centroid point of each cluster, the horizontal coverage radius of each DBS will be the horizontal distance from the centroid point to the furthest UE of each cluster. However, the system using such a placement recommendation will cause a large overlapping coverage area. If the overlapping coverage area becomes higher, it means that the distances between different deployed DBSs become short. Such a placement of DBSs may leads serious co-channel interference between the DBSs.

To alleviate the effect of the co-channel interference, the first task of this stage, placement refinement, will be used to refine the 2D location and coverage radius of each DBS. The procedure of the placement refinement solves the minimum enclosing circle problem~\cite{10.1007/BFb0038202} in linear time. After obtaining the minimum enclosing circle of each cluster, the system recognized it as the candidate coverage of each DBS and the center of each minimum enclosing circle will be the 2D candidate location of each DBS.

The system then computes and records the SINR value of each UE using the information of candidate coverage and 2D candidate location of each DBS. Since the above balanced $k$-means clustering and refinement does not handle the communications constraints yet, we need to check whether the demand of each clustered UE on data rates can be satisfied in this stage. 
If not, it means that some UEs are too far away from its associated DBS and the SINR of received signals can not exceed the threshold. In such a case, this kind of UEs may be re-associated with another nearby DBS and then get the satisfied data rate. Hence, the second task of this stage is to check the communications constraints of each UE and re-associate all the unsatisfied UEs. The system iteratively runs above operations of re-association and placement refinement until the all the clusters do not change. During the above loop, the horizontal locations and coverage radii of DBSs are updated iteratively. In the meanwhile, the system also uses~\eqref{eq:optimal_altitude_1} and~\eqref{eq:optimal_altitude_2} with the obtained coverage radii, the predefined allowable path-loss, and transmit power to derive the corresponding altitudes of DBSs.

The last task of this stage is to judge whether the obtained candidate placement is valid by checking the satisfaction rate. If the placement can not meet the given threshold of satisfaction rate, it means that the obtained candidate placement is invalid and the value of $k$ may be too small to satisfy the UE demand in the considered scenario.
The system will thereby do the whole procedure of this stage repeatedly with $k=k+1$ until the obtained candidate placement is valid. In fact, such a iterative search will stop by condition~\eqref{def:p3:c2}, which guarantee that the ratio of satisfied UEs to the all UEs exceeds the given satisfaction rate $\tau$.
\begin{figure*}[!t]
	\centering
	\begin{minipage}[t]{0.32\linewidth} 
		\centering 
		\includegraphics[width=\linewidth]{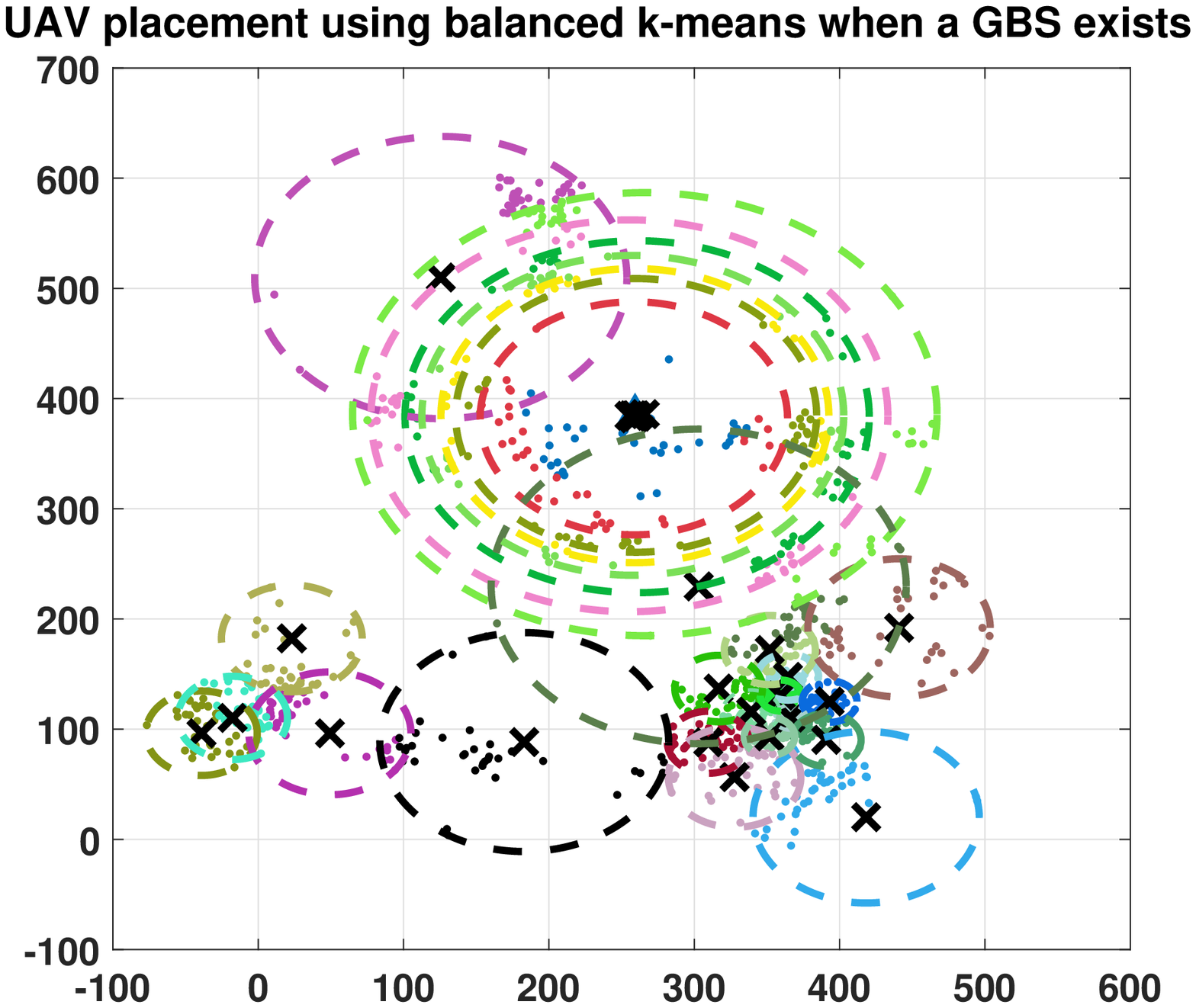} 
	\end{minipage}\hspace{.01\linewidth} 
	\begin{minipage}[t]{0.32\linewidth} 
		\centering 
		\includegraphics[width=\linewidth]{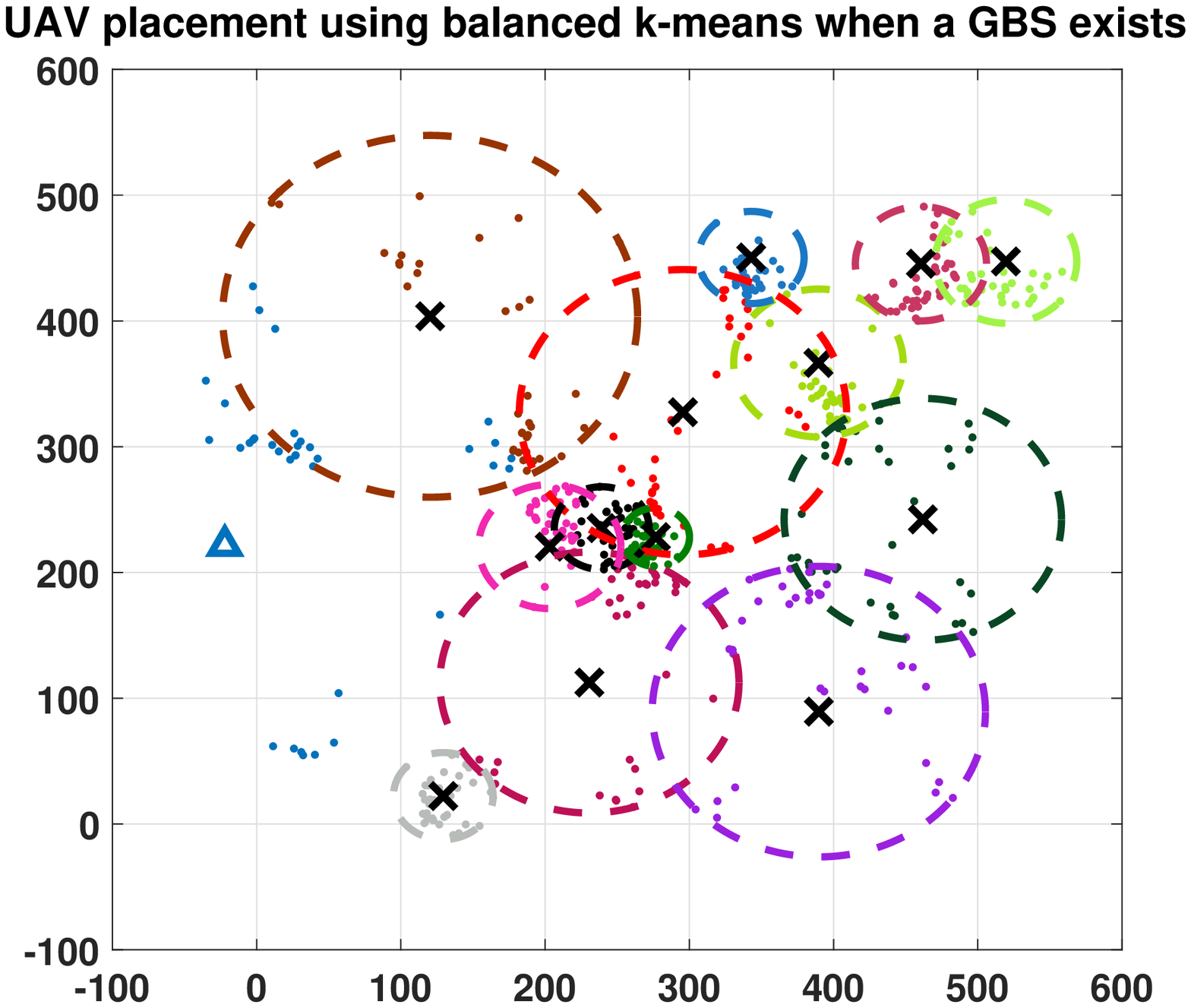} 
	\end{minipage}\hspace{.01\linewidth} 
	\begin{minipage}[t]{0.32\linewidth} 
		\centering 
		\includegraphics[width=\linewidth]{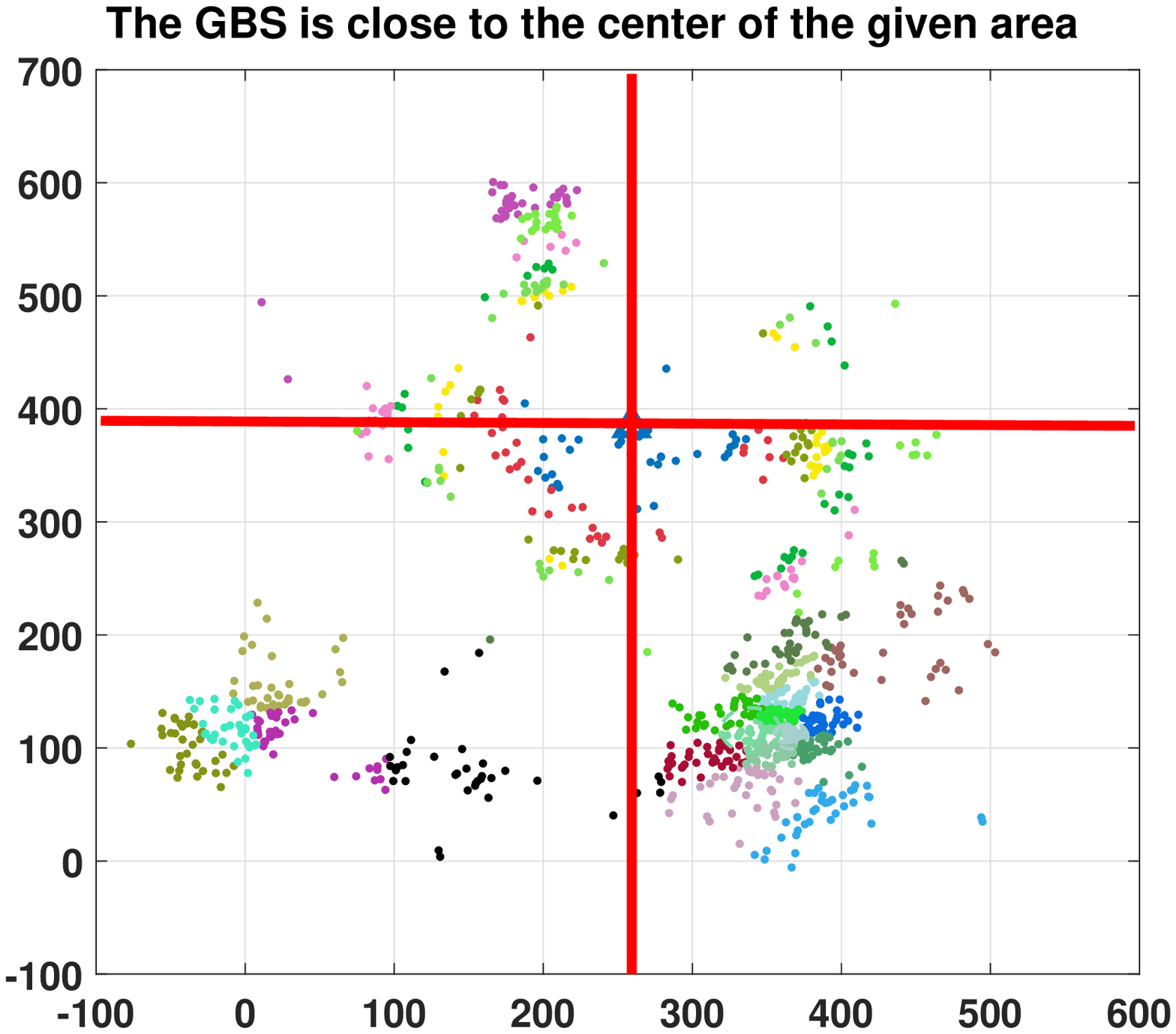} 
	\end{minipage}\\[-10pt]%
	\begin{minipage}[t]{0.32\linewidth}
		\caption{The coverage overlapping problem of balanced $k$-means if the location of the GBS is close to the center of the considered area.} 
		\label{fig:limitation:proposed}  
	\end{minipage}\hspace{.01\linewidth} 
	\begin{minipage}[t]{0.32\linewidth} 
		\caption{No coverage overlapping problem occurs if the location of the GBS is close to the boundary of the considered area.} 
		\label{fig:limitation:compared} 
	\end{minipage}\hspace{.01\linewidth} 
	\begin{minipage}[t]{0.32\linewidth} 
		\caption{Use the coordinates of the GBS to pre-partition the considered area for alleviating the coverage overlapping problem.} 
		\label{fig:limitation:enhancement} 
	\end{minipage}%
\end{figure*}

\section{The Enhanced DDP (eDDP)}
\label{sec:eddp}
Although the proposed DDP can provide a balanced placement result that each DBS serves almost the same number of UEs, a new issue, coverage overlapping problem, arises when we try to cluster UEs when some UEs already have associated with the GBS.  Fig.~\ref{fig:limitation:proposed} shows the result of overlapping placement case. The coverage overlapping problem leads significant interference, and thus reduces system sum rate. The reason is that the conventional approaches does not consider the altitude of a DBS which affects the associations between UEs and DBSs. Furthermore, most existing works do not consider the placement problem when a GBS exists. In short, the enhancement or adaptation of conventional clustering approach for the considered problem is required.

To mitigate the negative impact of the coverage overlapping problem, we conduct the multiple clustering tests with different input data (locations of UEs and the GBS) and observe the output clustered results. As shown in Fig.~\ref{fig:limitation:compared}, we find that the coverage overlapping problem rarely occurs when the location of GBS is close to the boundary of given area. We thus propose an enhancement to reduce the occurrence the coverage overlapping problem. The enhancement is to pre-partition the considered area into two sub-regions, as shown in Fig.~\ref{fig:limitation:enhancement}. The system considers the two sub-areas obtained as two separate inputs and then runs DDP (Algorithm~\ref{alg:ddp}) to find their placement recommendations separately. To make the system can work automatically, a pre-partition rule must be defined in advance. The system can follow the rules to trigger the pre-partition, and thus reduce the occurrence of coverage overlapping problem.

Considering the ideal case without interference, we can obtain the effective coverage radius (or propagation distance) $r_G$ of the GBS according to~\eqref{eq:sinr:g2u} and the given SINR threshold $\gamma_{\text{th}}$. Then, we use the location $L_G$ and ideal coverage radius $r_G$ of the GBS to define the pre-partition rule. Given a target area $\mathcal{A}=\{[X_{\min}^{\mathcal{A}},X_{\max}^{\mathcal{A}}],[Y_{\min}^{\mathcal{A}},Y_{\max}^{\mathcal{A}}]\}$ where $[X_{\min}^{\mathcal{A}},X_{\max}^{\mathcal{A}}]$ and $[Y_{\min}^{\mathcal{A}},Y_{\max}^{\mathcal{A}}]$ are $X$-axis and $Y$-axis boundaries of $\mathcal{A}$, let $\text{dist}(L_G,X_{\min}^{\mathcal{A}})$, $\text{dist}(L_G,X_{\max}^{\mathcal{A}})$, $\text{dist}(L_G,Y_{\min}^{\mathcal{A}})$, and $\text{dist}(L_G,Y_{\max}^{\mathcal{A}})$ are the distances from the GBS to different boundaries, symbols ``$\wedge$" and ``$\vee$" are respectively the logical operators ``and" and ``or", the pre-partition rule are defined as
\begin{enumerate}
	\item If $(\text{dist}(L_G,X_{\min}^{\mathcal{A}})>r_G\wedge \text{dist}(L_G,X_{\max}^{\mathcal{A}})>r_G)\wedge(\text{dist}(L_G,Y_{\min}^{\mathcal{A}})>r_G\wedge \text{dist}(L_G,Y_{\max}^{\mathcal{A}})>r_G)$, the system uses the location of the GBS, $L_G$, to partition $\mathcal{A}$ into four sub-regions;
	\item Else if $(\text{dist}(L_G,X_{\min}^{\mathcal{A}})>r_G\wedge \text{dist}(L_G,X_{\max}^{\mathcal{A}})>r_G)\wedge(\text{dist}(L_G,Y_{\min}^{\mathcal{A}})\leq r_G\vee \text{dist}(L_G,Y_{\max}^{\mathcal{A}})\leq r_G)$, use the $X$-axis coordinate of the GBS, $L_G.x$, to partition $\mathcal{A}$ into two sub-regions.
	\item Else if $(\text{dist}(L_G,X_{\min}^{\mathcal{A}})\leq r_G\vee \text{dist}(L_G,X_{\max}^{\mathcal{A}})\leq r_G)\wedge(\text{dist}(L_G,Y_{\min}^{\mathcal{A}})> r_G\wedge \text{dist}(L_G,Y_{\max}^{\mathcal{A}})> r_G)$, use the $Y$-axis coordinate of the GBS, $L_G.y$, to partition $\mathcal{A}$ into two sub-regions.
	\item Otherwise, do nothing.
\end{enumerate}

Although the occurrence of the coverage overlapping problem can be reduced by the pre-partition, some additional computation needs to be processed. Some DBSs deployed in different sub-regions are close to the partition line and their coverage may overlap each other, the obtained local SINR information of each UE is not the final result. That is, after executing DDP for the different sub-regions, the system runs a global integration to update the SINR information of the UEs who covered by two DBSs in different sub-regions. Finally, the system can computes the correct system sum rate. Overall, the pre-partition rule is depicted in Fig~\ref{fig:partition_region} and the procedure of enhanced DDP (eDDP) is summarized as Algorithm~\ref{alg:eddp}.

\begin{figure}[!t]
	\centering
	\includegraphics[width=.75\columnwidth]{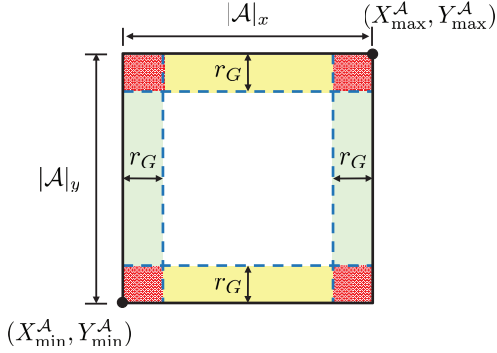}
	\caption{If the GBS is in the (center) white region, the serving area will be partitioned into 4 sub-regions using $x=L_G.x$ and $y=L_G.y$; if the GBS is in the (top or down) yellow regions, the serving area will be partitioned into 2 sub-regions using $x=L_G.x$; if the GBS is in the (left or right) green regions, the serving area will be partitioned into 2 sub-regions using $y=L_G.y$; otherwise, do nothing.}
	\label{fig:partition_region}
\end{figure}

\section{Complexity Analysis}
\label{sec:analysis}
In this part, we discuss the complexity of average case for the proposed algorithms. 
In the first stage of DDP, the system just reads the input information and determine the initial value of $k$ by~\eqref{eq:initial_value_of_k}. The above operations cost $N+\sigma\approx\mathcal{O}(N)$ time where $N$ is the number of UEs and $\sigma$ is the constant time for initializing the used variables and $k$. Note that $k$-means based clustering is proven to be NP-hard in~\cite{MAHAJAN201213}, which implies that the clustering problem in the second stage is also NP-hard. In general, $k$-means clustering is an efficient way to find a near optimal (or sub-optimal) solution and the complexity of $k$-means is $\mathcal{O}(Ndkt)$, where the dimensionality is $d=2$ in our work and $t$ is the number of iterations for $k$-means to converge. In average case, $d$ is constant and $t$ is very small, so the complexity of $k$-means can approximate $\mathcal{O}(Ndkt)\approx\mathcal{O}(kN)$. 
In the second stage of DDP, we adopt the balanced $k$-means clustering~\cite{10.1007/978-3-662-44415-3_4} for fairly associating UEs with DBSs. The assignment step of the balanced $k$-means clustering can be solved by the \emph{Hungarian} algorithm~\cite{doi:10.1137/0105003}, so the time complexity of the second stage is $\mathcal{O}(tN^3)\approx\mathcal{O}(N^3)$.

The last stage of DDP includes two tasks: UE re-association and placement refinement. Suppose that $\bar{n}_\text{re}$ is the average number of unsatisfied UEs which needs to be re-associated and each unsatisfied UE may be re-associated $k-1$ times in the worst case, the UE re-association task costs about $(k-1)\bar{n}_\text{re}\approx\mathcal{O}(k\bar{n}_\text{re})$ time. For the placement refinement, the system can refine the placement of each DBS by solving the minimum enclosing circle problem in linear time, so the time complexity of refining the placement of $k$ DBS will be $\sum_{j=1}^{k}N_j=\mathcal{O}(kN)$.
Overall, executing the above three-stage operations one time costs $\mathcal{O}(N+N^3+k\bar{n}_\text{re}+kN)\approx\mathcal{O}(N^3)$ time since $N\gg \bar{n}_\text{re}$ and $N\gg k$ in general. If the system directly searches for the placement result by increasing $k$ from $1$, the total time complexity will be $\mathcal{O}(kN^3)$. In the proposed DDP, the system searches the placement result starting within the initial lower bound value of $k$, $k_\text{min}$, by~\eqref{eq:initial_value_of_k} and the final actual value of $k$, $k_\text{actual}$. Hence, the time complexity of DDP is $\mathcal{O}(\Delta_k N^3)$ where $\Delta_k=k_\text{actual}-k_\text{min}+1$. 

\begin{algorithm2e}[!t]
	\footnotesize
	\SetAlgoLined
	\KwIn{target area $\mathcal{A}$, dataset of UE locations $L_E$, location of the GBS $L_G$, the maximum number of DBS $K$, SINR threshold $\gamma_{\text{th}}$, the channel bandwidth provided by each DBS and the GBS $B$, the transmit power of the GBS $P_G$, the transmit power of a DBS $P_{\text{drone}}$, and the minimum data rate requirement $c_{\min}$}
	\KwOut{association information $L_{\text{association}}$, DBS locations $L_{\text{drone}}^{\text{cand}}$, and DBS altitudes $L_{\text{altitude}}$}
	\tcc{Use the same variables in Algorithm~\ref{alg:ddp}}
	let $L_p$ to store the partition line\;	
	parse the input information of target area $\mathcal{A}$ and get the distance relations between the GBS and the boundaries\;
	\uIf{$(\text{dist}(L_G,X_{\min}^{\mathcal{A}})>r_G\wedge \text{dist}(L_G,X_{\max}^{\mathcal{A}})>r_G)\wedge(\text{dist}(L_G,Y_{\min}^{\mathcal{A}})>r_G\wedge \text{dist}(L_G,Y_{\max}^{\mathcal{A}})>r_G)$}{\label{alg:eddp:prepartition:start}
		use the lines $x=L_G.x$ and $y=L_G.y$ to partition $\mathcal{A}$ into four sub-regions\;
	}
	\uElseIf{$(\text{dist}(L_G,X_{\min}^{\mathcal{A}})>r_G\wedge \text{dist}(L_G,X_{\max}^{\mathcal{A}})>r_G)\wedge(\text{dist}(L_G,Y_{\min}^{\mathcal{A}})\leq r_G\vee \text{dist}(L_G,Y_{\max}^{\mathcal{A}})\leq r_G)$}{
		use the line $x=L_G.x$ to partition $\mathcal{A}$ into two sub-regions\;
	}
	\uElseIf{$(\text{dist}(L_G,X_{\min}^{\mathcal{A}})\leq r_G\vee \text{dist}(L_G,X_{\max}^{\mathcal{A}})\leq r_G)\wedge(\text{dist}(L_G,Y_{\min}^{\mathcal{A}})> r_G\wedge \text{dist}(L_G,Y_{\max}^{\mathcal{A}})> r_G)$}{
		use the line $y=L_G.y$ to partition $\mathcal{A}$ into two sub-regions\label{alg:eddp:prepartition:end}\;
	}
	\Else{
		\tcc{Normal DDP witout the partition}
		run Algorithm~\ref{alg:ddp} with $\mathcal{A}$ directly and then jump to line~\ref{alg:eddp:return}\label{alg:eddp:call_ddp:normal}\;
	}
	launch two/four threads to run Algorithm~\ref{alg:ddp} to find the placements for the obtained two/four sub-regions in parallel\;
	merge all of the obtained output information in $L_{\text{association}}$,$L_{\text{drone}}^{\text{cand}}$, $L_{\text{SINR}}$, $L_{\text{radius}}$, and $L_{\text{altitude}}$\label{alg:eddp:update_sinr:start}\;
	find the DBSs whose coverage circles are intersects the partition lines and save the result in a temporary list $L_{\text{overlapped}}^{\text{possible}}$\;
	find the UEs covered by multiple DBSs which are in $L_{\text{overlapped}}^{\text{possible}}$ and then update these UEs' SINR in $L_{\text{SINR}}$\label{alg:eddp:update_sinr:end}\;
	\Return $L_{\text{association}}$,$L_{\text{drone}}^{\text{cand}}$, $L_{\text{SINR}}$, $L_{\text{radius}}$, and $L_{\text{altitude}}$
	\label{alg:eddp:return}\;
	\caption{The procedure of eDDP}
	\label{alg:eddp}
\end{algorithm2e}

In eDDP, the complexity of determining the rule condition is $\mathcal{O}(1)$. The pre-partition operations from lines~\ref{alg:eddp:prepartition:start} to~\ref{alg:eddp:prepartition:end} of Algorithm~\ref{alg:eddp} cost about $\mathcal{O}(N)$ since the system needs to check the locations of all the UEs. The final update of UEs' SINR is performed by the operations from lines~\ref{alg:eddp:update_sinr:start} to~\ref{alg:eddp:update_sinr:end} and the complexity is $\mathcal{O}(N)$. 
Suppose the number of partition is $n_\text{partition}\in\{1,2,4\}$ in our design, the worst case of eDDP occurs when $n_\text{partition}=1$ and the complexity of eDDP which is executed at line~\ref{alg:eddp:call_ddp:normal} will be as same as DDP. 
Although both DDP and eDDP have the same time complexity in worst case, eDDP has better performance in average time complexity compared to DDP. The reduction of average time complexity be concluded as following theorem.
\begin{thm}[Average Complexity Reduction]
	Suppose that the notations are defined as above and the location of the GBS is given randomly in the target rectangle area $\mathcal{A}$, compared to DDP, eDDP effectively reduces the average time complexity by
	\begin{align}
	\left(1-\frac{(|\mathcal{A}|_x+14r_G)\times(|\mathcal{A}|_y+14r_G)}{64\times|\mathcal{A}|_x\times|\mathcal{A}|_y}\right)\times 100\%,
	\end{align}
	where $|\mathcal{A}|_x$ and $|\mathcal{A}|_y$ are the length and the width of $\mathcal{A}$.
\end{thm}
\begin{IEEEproof}
	In general, the time complexity of eDDP can be expressed as $\mathcal{O}(\Delta_k (N/n_\text{partition})^3)$. On one hand, the best case of eDDP occurs when $n_\text{partition}=4$ and each partition has almost the same number of UEs, and thus the best time complexity will be $\frac{1}{64}\mathcal{O}(\Delta_k N^3)$. On the other hand, in the case of $n_\text{partition}=2$, eDDP will cost $\frac{1}{8}\mathcal{O}(\Delta_k N^3)$ time. However, different partition cases have different occurrence probabilities according to the proposed pre-partition rule of eDDP. The occurrence probability of each partition case can be obtained by
	\begin{align*}
	\text{Pr}\left(n_\text{partition}=1\right)&=\dfrac{4r_G^2}{|\mathcal{A}|_x\times|\mathcal{A}|_y},\\
	\text{Pr}\left(n_\text{partition}=2\right)&=\dfrac{2r_G\times\left(|\mathcal{A}|_x+ |\mathcal{A}|_y-4r_G\right)}{|\mathcal{A}|_x\times|\mathcal{A}|_y},\\
	\text{Pr}\left(n_\text{partition}=4\right)&=\dfrac{(|\mathcal{A}|_x-2r_G)\times \left(|\mathcal{A}|_y-2r_G\right)}{|\mathcal{A}|_x\times|\mathcal{A}|_y},\text{ and}\\
	\text{Pr}\left(n_\text{partition}=1\right)&+\text{Pr}\left(n_\text{partition}=2\right)+\text{Pr}\left(n_\text{partition}=4\right)=1.
	\end{align*}
	
	\noindent
	Then, the average time complexity of eDDP will be
	\begin{align*}
	T^\text{eDDP}_\text{average}=&\text{Pr}\left(n_\text{partition}=1\right)\times\mathcal{O}(\Delta_k N^3)\\
	&+	\text{Pr}\left(n_\text{partition}=2\right)\times\dfrac{1}{8}\mathcal{O}(\Delta_k N^3)\\
	&+\text{Pr}\left(n_\text{partition}=4\right)\times\dfrac{1}{64}\mathcal{O}(\Delta_k N^3)\\
	=&\dfrac{(|\mathcal{A}|_x+14r_G)\times(|\mathcal{A}|_y+14r_G)}{64\times|\mathcal{A}|_x\times|\mathcal{A}|_y}\times\mathcal{O}(\Delta_k N^3).
	\end{align*}
	
	\noindent
	Since the average time complexity of DDP is $T^\text{DDP}_\text{average}=\mathcal{O}(\Delta_k N^3)$ and it is easy to see that $\frac{(|\mathcal{A}|_x+14r_G)\times(|\mathcal{A}|_y+14r_G)}{64\times|\mathcal{A}|_x\times|\mathcal{A}|_y}< 1$, the reduction of average time complexity by eDDP can be calculated as
	\begin{align*}
	\dfrac{T^\text{DDP}_\text{average}-T^\text{eDDP}_\text{average}}{T^\text{DDP}_\text{average}}&\times 100\%= \\
	&\hspace{-5em}\left(1-\frac{(|\mathcal{A}|_x+14r_G)\times(|\mathcal{A}|_y+14r_G)}{64\times|\mathcal{A}|_x\times|\mathcal{A}|_y}\right)\times 100\%.
	\end{align*}
\end{IEEEproof}


\vspace{-10pt}
\section{Simulations}
\label{sec:simulation}
The simulation including all compared approaches are implemented in MATLAB R2017b. The simulation program is executed on a Windows 10 server with an Intel(R) Core(TM) i7-7700 CPU @ 3.60GHz and 8GB $\times$ 2 memory. Since the proposed approaches, DDP and eDDP, are modified from $k$-means based clustering, thus we compare our algorithms with balanced $k$-means~\cite{10.1007/978-3-662-44415-3_4} approach in the simulation. All comparison approaches are briefly described below.
\begin{enumerate}
	\item Balanced $k$-means provides a balanced clustering recommendation based on the horizontal distance between UE and DBS.  
	\item The proposed DDP first performs balanced $k$-means clustering and then executes the proposed re-association stage for DBS placement refinement considering the SINR of each UE.
	\item The proposed eDDP partitions the target serving area according to the proposed rules before executing DDP, and thus reduce the occurrence of the coverage overlapping problem.
\end{enumerate}

We use different artificial datasets as the input spatial information and these datasets contain different number of UE locations, $N=\{400,500,\dots,800\}$, which are arbitrarily distributed over a $600\times 600$ m$^2$ area. We consider the urban scenario and its environmental parameters are $(a,b,\eta_{LoS},\eta_{NLoS}) = (9.61,0.16,1,20)$ given by~\cite{7744808}. We assume the maximum number of available DBSs is $K=100$, the maximum allowable path-loss of GBS-to-DBS and DBS-to-UE links is $L_{\text{allowable}}=119$ (dB), and the path-loss exponent of downlink from GBS to UE is $\alpha = 6.5$. For DDP and eDDP, the default target sanctification rate of UEs is set to $0.4$, which means that the system finds a placement result to guarantee more than $40\%$ of the UEs have satisfied SINR and data rates.
The detailed important simulation parameters and predefined constraints are presented in Table~\ref{simulation_parameters}. 

\begin{table}[!t]
	\footnotesize
	\renewcommand{\arraystretch}{1.1}
	\caption{Simulation Parameters}
	\label{simulation_parameters}
	\centering
	\begin{tabular}{p{4.7cm}ll}
		\hline
		\textbf{Parameter} & \textbf{Symbol} & \textbf{Value}\\
		\hline 
		Cellular transmit power of the GBS & $P_G$ & $40$ dBm \\
		mmWave transmit power of the GBS & $P_G^{\text{bk}}$ & $30$ dBm\\
		Transmit power of each DBS & $P_j$ & $20$ dBm \\
		Maximum allowable path-loss of GBS-to-DBS and DBS-to-UE links & $L_{\text{allowable}}$ & $119$ \\
		Path-loss exponent of the GBS-to-UE downlink & $\alpha$ & $6.5$ \\
		Minimum data rate requirement of each UE & $c_{\min}$ & $1$ Mbps \\
		The number of UEs & $N$ & $400$ to $800$ \\
		Maximum number of DBSs & $K$ & $100$ \\
		Target satisfaction rate of UEs (Default) & $\tau$ & $40\%$ \\
		Cellular thermal noise power spectral density & $N_{0}$ & $-174$ dBm/Hz \\
		Cellular carrier frequency & $f_c$ & $2$ GHz \\
		mmWave carrier frequency & $f_c^\text{bk}$ & $28$ GHz \\	
		Cellular carrier bandwidth & $B$ & $20$ MHz \\
		mmWave carrier bandwidth & $B^\text{bk}$ & $20\times100$ MHz \\	
		Cellular SINR threshold & $\gamma_{\text{th}}$ & $5$ dB \\ 
		mmWave SINR threshold & $\gamma_{\text{th}}^\text{bk}$ & $-10$ dB \\	
		Minimum altitude of each DBS & $h_{\min}$ & $20$ m \\
		Maximum altitude of each DBS & $h_{\max}$ & $400$ m \\	
		\hline
	\end{tabular}
\end{table}
\begin{figure*}[!ht]
	\centering
	\subfigure[Balanced $k$-means]{
		\label{fig:placement_result:bkm} 
		\includegraphics[width=0.33 \textwidth]{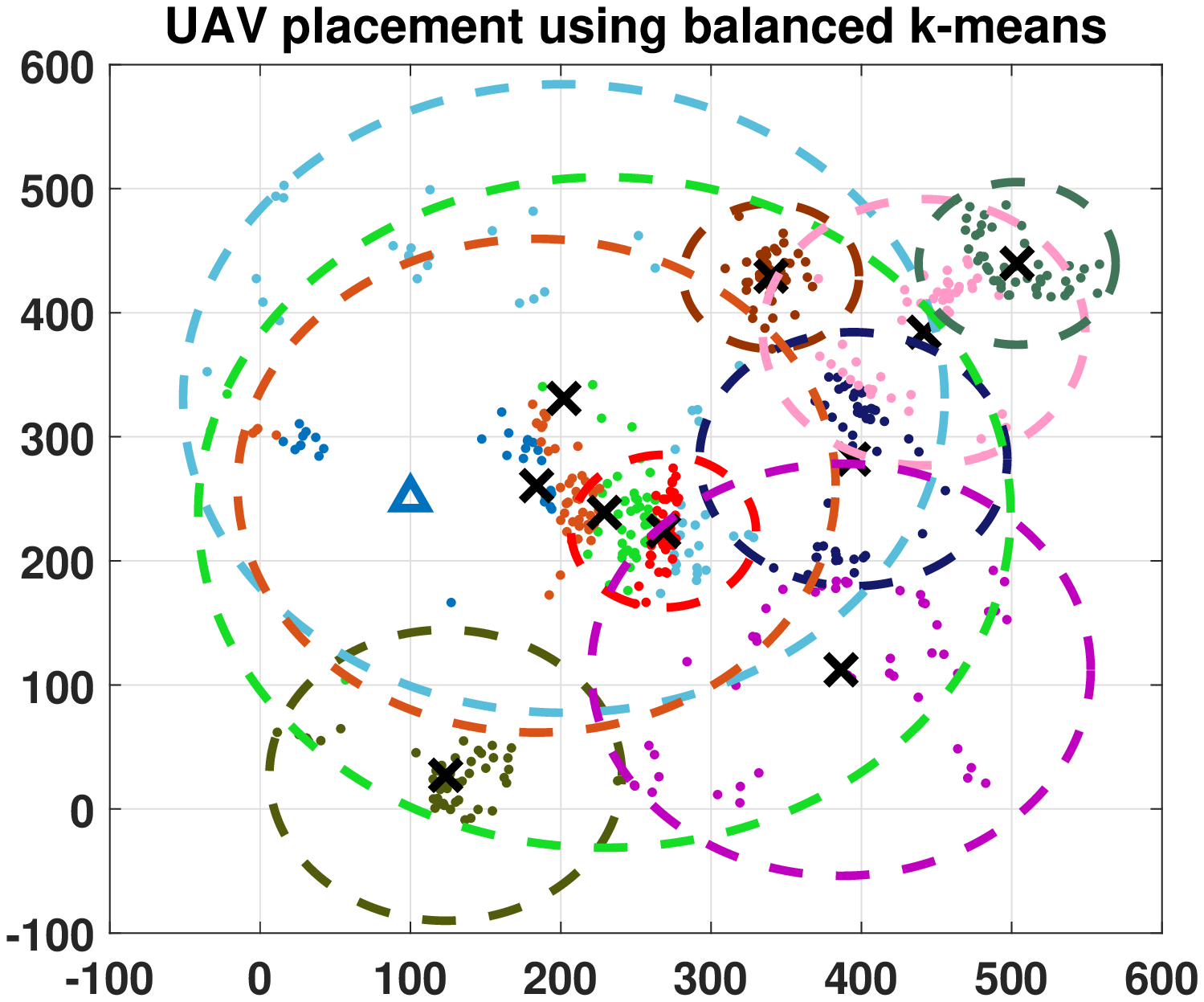}}%
	\subfigure[DDP]{
		\label{fig:placement_result:ddp} 
		\includegraphics[width=0.33 \textwidth]{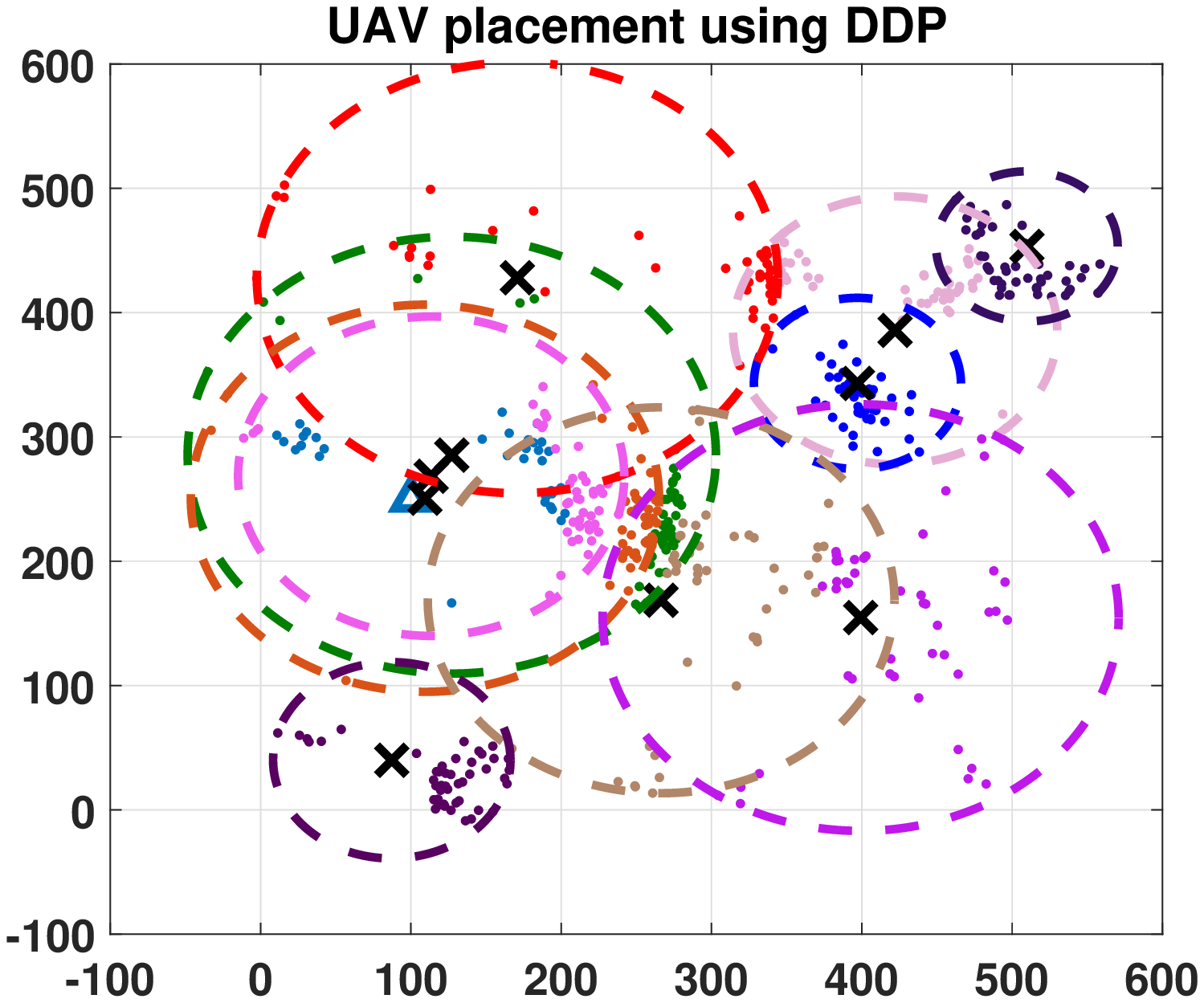}}%
	\subfigure[eDDP]{
		\label{fig:placement_result:eddp} 
		\includegraphics[width=0.33 \textwidth]{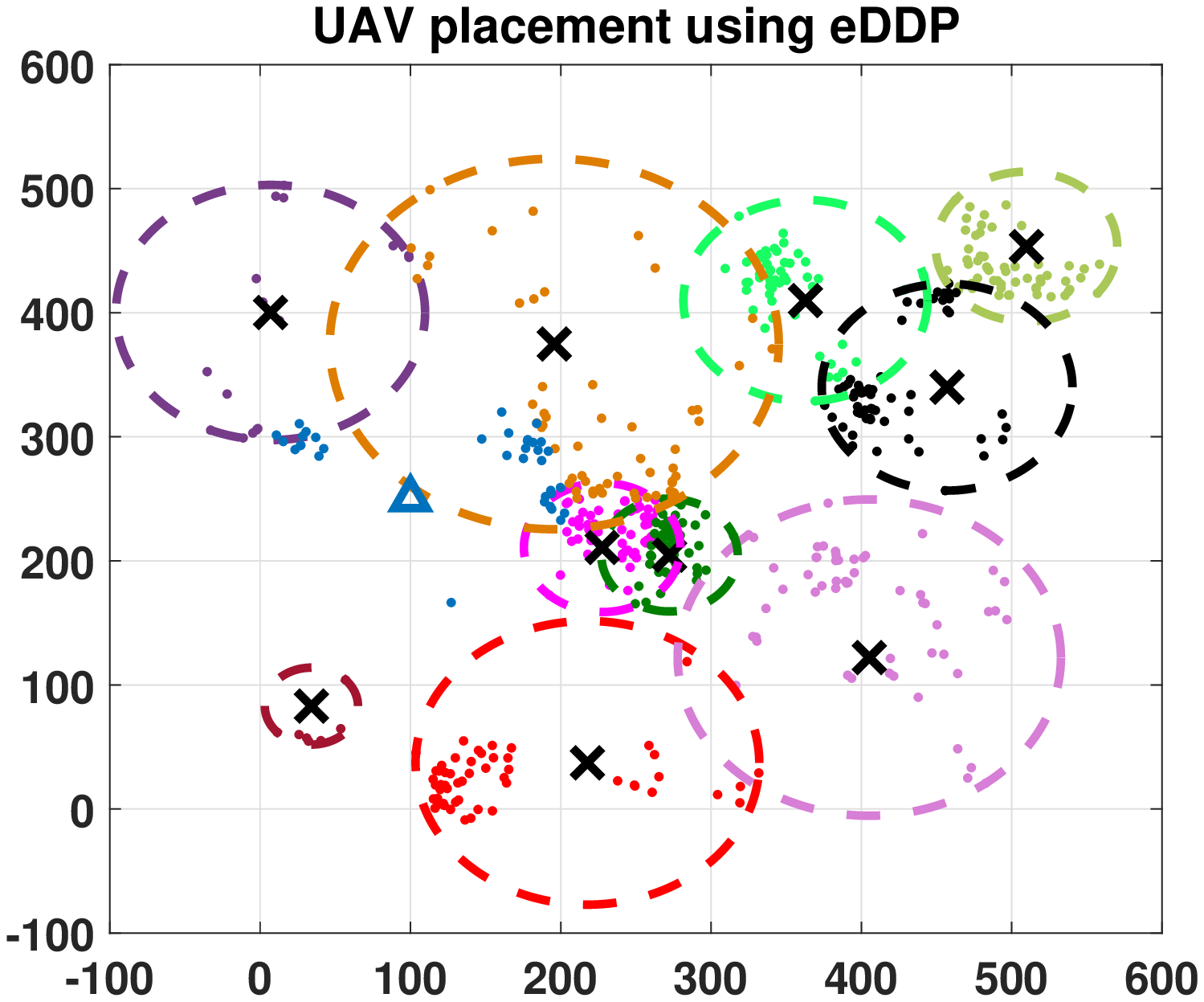}}%
	\caption{The placement results of 
		\subref{fig:placement_result:bkm} balanced $k$-means, \subref{fig:placement_result:ddp} DDP, and \subref{fig:placement_result:eddp} eDDP under the same UE distribution in a urban scenario when $N=500$, $k=10$, and the GBS locates as $(100,250)$. 
	}
	\label{fig:placement_result} 
\end{figure*}

\subsection{Placement Results}
In the first simulation, we choose one of input location data set ($N=500$) to show the placement results of the compared balanced $k$-means approach, DDP, and eDDP in Fig.~\ref{fig:placement_result}. In this scenario, the GBS locates at $(100,250)$ and some very high dense flash crowd events occur around the following coordinates: $(200,250)$, $(150,20)$, $(340,430)$, $(400,340)$, and $(480,430)$. In Fig.~\ref{fig:placement_result}, the white triangle is the GBS, black crosses are DBSs, small dots are UEs, and each dashed-circle is the coverage of the corresponding DBS.

\begin{figure*}[!t]
	\centering
	\subfigure[Empirical Probability Distribution Function]{
		\label{fig:satisfaction_rate_of_methods:pdf} 
		\includegraphics[width=0.35 \textwidth]{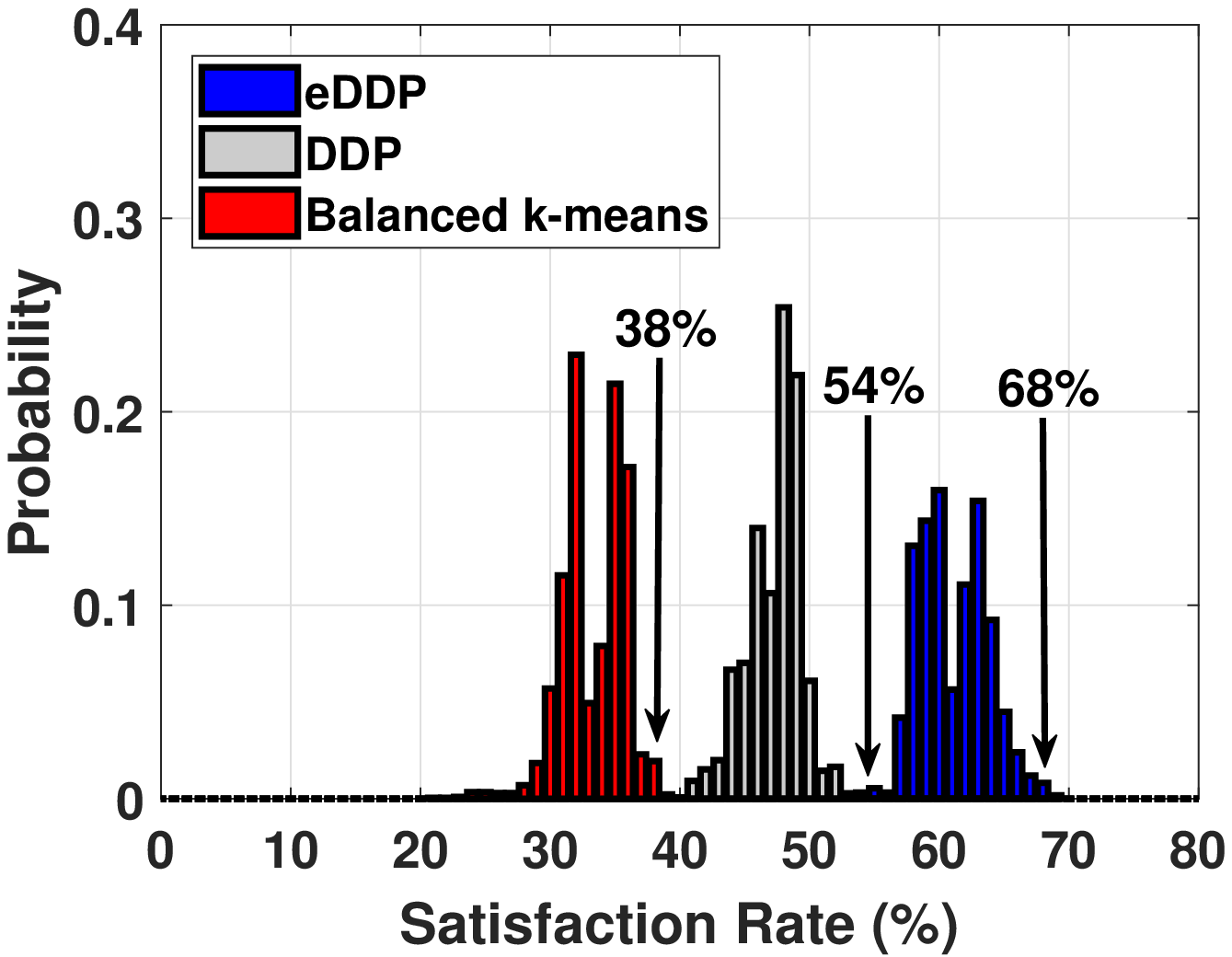}}\hskip 20pt%
	\subfigure[Empirical Cumulative Distribution Function]{
		\label{fig:satisfaction_rate_of_methods:cdf} 
		\includegraphics[width=0.35 \textwidth]{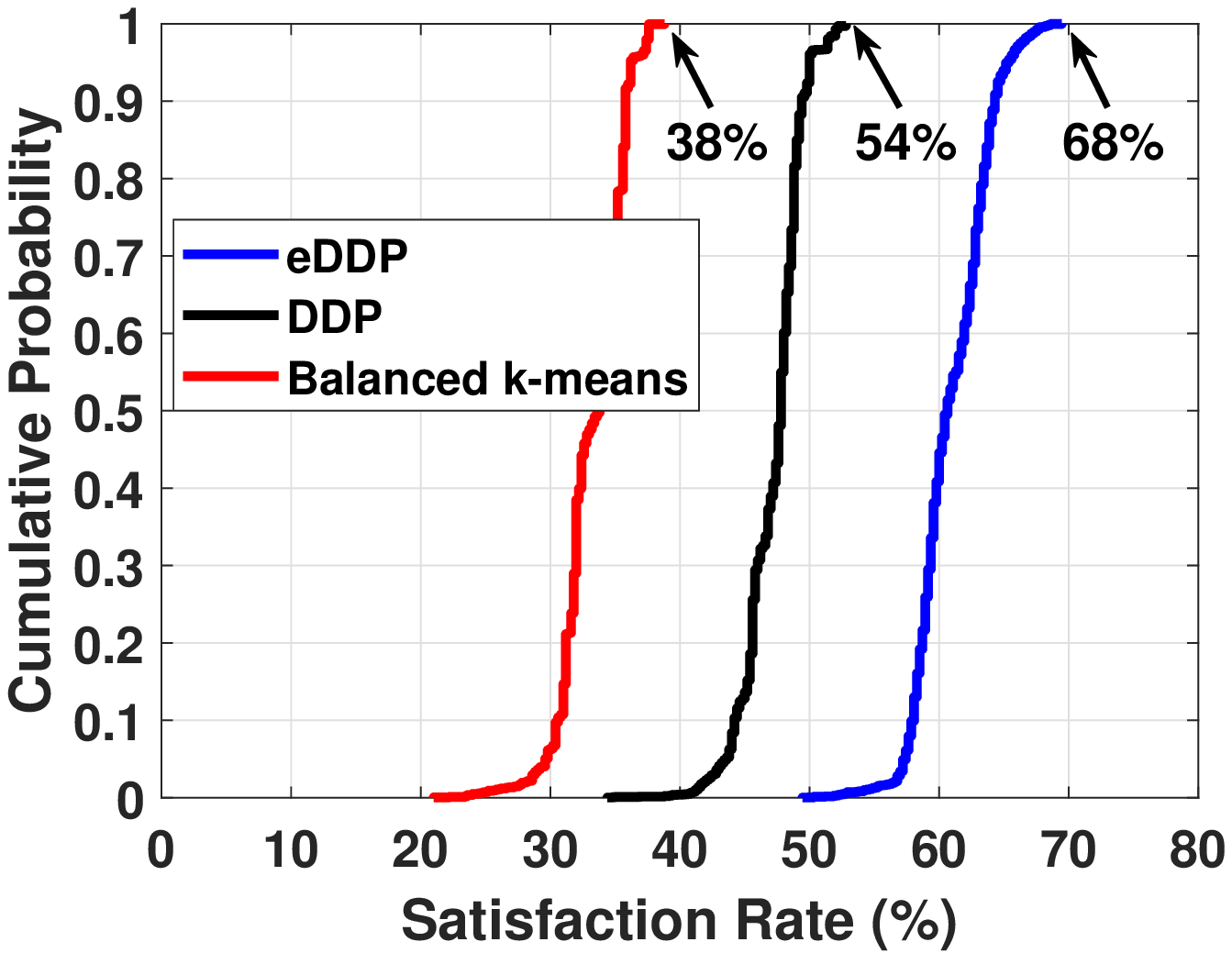}}%
	\caption{The performance results of different approaches in terms of~\subref{fig:satisfaction_rate_of_methods:pdf} Empirical PDF and~\subref{fig:satisfaction_rate_of_methods:cdf} Empirical CDF on the satisfaction rate of UEs when $N=500$, $k=10$, and the GBS locates at $(100,250)$.
	}
	\label{fig:satisfaction_rate_of_methods} 
\end{figure*}

Unlike the conventional $k$-means based approach which cannot determine the value of $k$ by the algorithms themselves, the proposed DDP and eDDP can automatically determine the initial lower bound value of $k$ by~\eqref{eq:initial_value_of_k} with the input spatial information, environmental parameters, and communications constraints. To observe the difference between the proposed solutions and balanced $k$-means, we directly use the same number ($k=10$) of DBSs as the one obtained by eDDP to get the placement results of balanced $k$-means in Fig.~\ref{fig:placement_result:bkm}.
Fig.~\ref{fig:placement_result:bkm} shows that the balanced $k$-means place $k=10$ DBSs to serve $N=500$ UEs in the target area and have very large overlapping coverage areas. 
Fig.~\ref{fig:placement_result:ddp} shows that DDP can significantly reduce the overlapping coverage area in comparison with balanced $k$-means. Such a result is benefited by the proposed placement refinement step. However, sometimes the coverage overlapping problem still occurs in both balanced $k$-means and DDP. According to the result in Fig.~\ref{fig:placement_result:eddp}, compared with DPP and balanced $k$-means, eDDP can alleviate the coverage overlapping problem and get the most effective placement with the smallest overlapping coverage area.

\subsection{The Target Satisfaction Rate of UEs}
In this subsection, to verify the effectiveness of each method, we performed each compared approach $10^4$ times using random initial positions of DBSs and recorded each deployment result and the satisfaction rate achieved in the case of $N=500$ and $k=10$. Through statistics, we can observe the \emph{empirical probability distribution function} (empirical PDF) and \emph{empirical cumulative distribution function} (empirical CDF) of the satisfaction rate of each method. According to the empirical PDF and empirical CDF of all methods in Fig.~\ref{fig:satisfaction_rate_of_methods:pdf} and Fig.~\ref{fig:satisfaction_rate_of_methods:cdf}, eDDP is the most effective solution and it can provide the best DBS placement result with $68\%$ satisfaction rate. DDP and balanced $k$-means can achieve $54\%$ and $38\%$ satisfaction rates, respectively. 

\begin{figure*}[!t]
	\centering
	\subfigure[]{
		\label{fig:ratio_of_satisfied_ues:time} 
		\includegraphics[width=0.33 \textwidth]{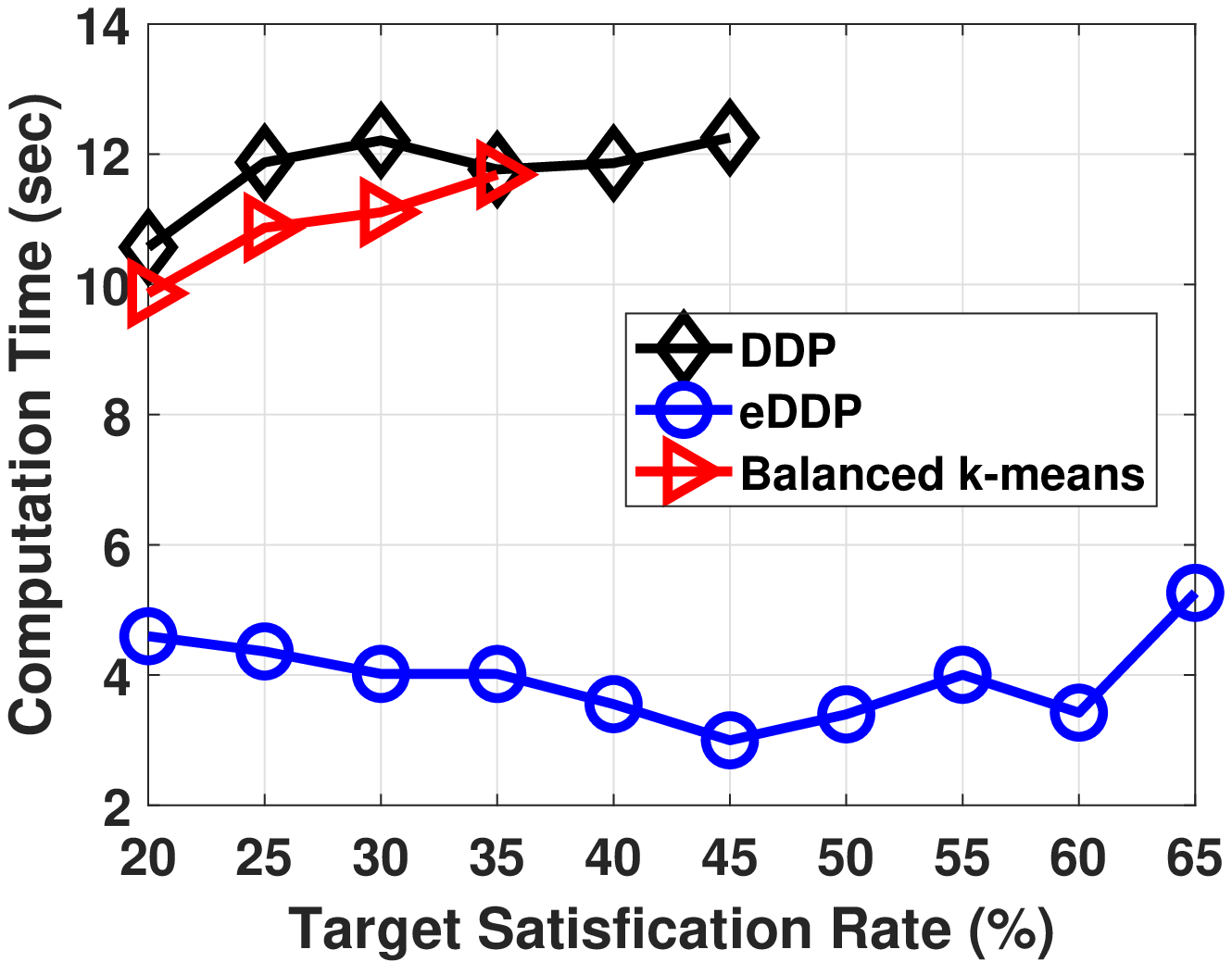}}%
	\subfigure[]{
		\label{fig:ratio_of_satisfied_ues:number_of_dbss} 
		\includegraphics[width=0.33 \textwidth]{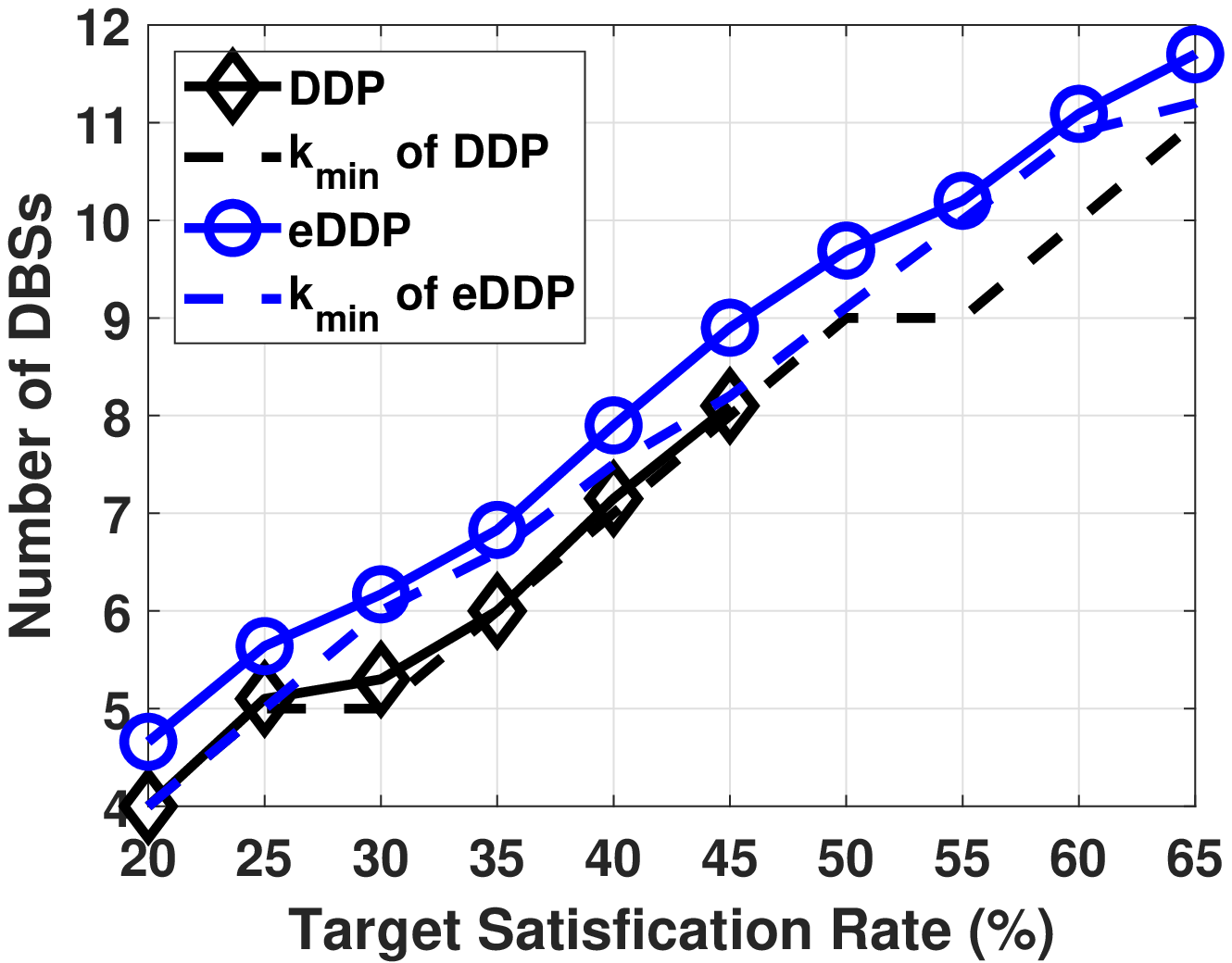}}%
	\subfigure[]{
		\label{fig:ratio_of_satisfied_ues:sumrate} 
		\includegraphics[width=0.33 \textwidth]{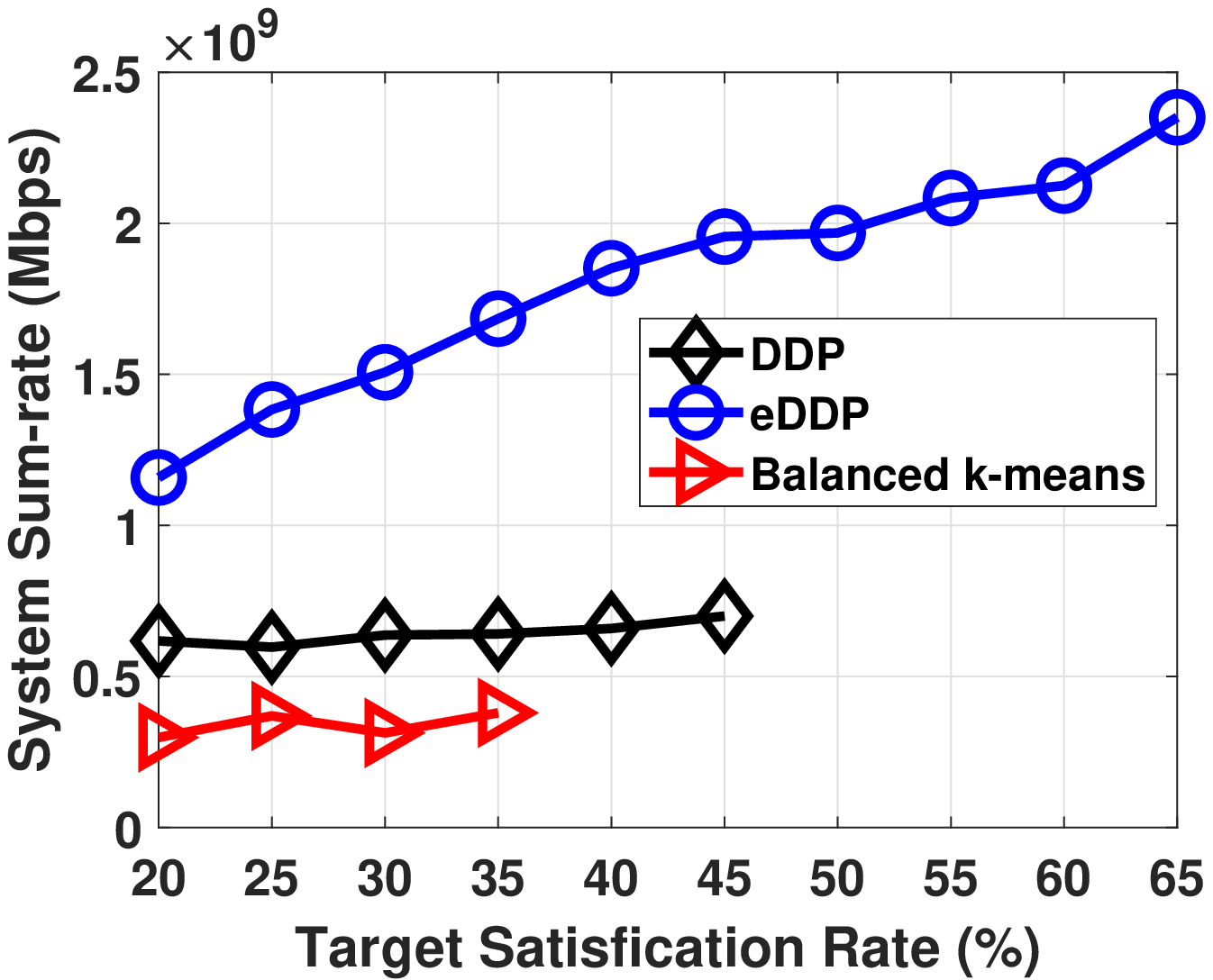}}%
	\caption{The performance results of different approaches in terms of~\subref{fig:ratio_of_satisfied_ues:time} computation time,~\subref{fig:ratio_of_satisfied_ues:number_of_dbss} number of DBSs, and~\subref{fig:ratio_of_satisfied_ues:sumrate} system sum rate while varying the target satisfaction rate of UEs when $N=500$.
	}
	\label{fig:ratio_of_satisfied_ues} 
\end{figure*}

\begin{figure*}[!t]
	\centering
	\subfigure[]{
		\label{fig:number_of_ues:time} 
		\includegraphics[width=0.33 \textwidth]{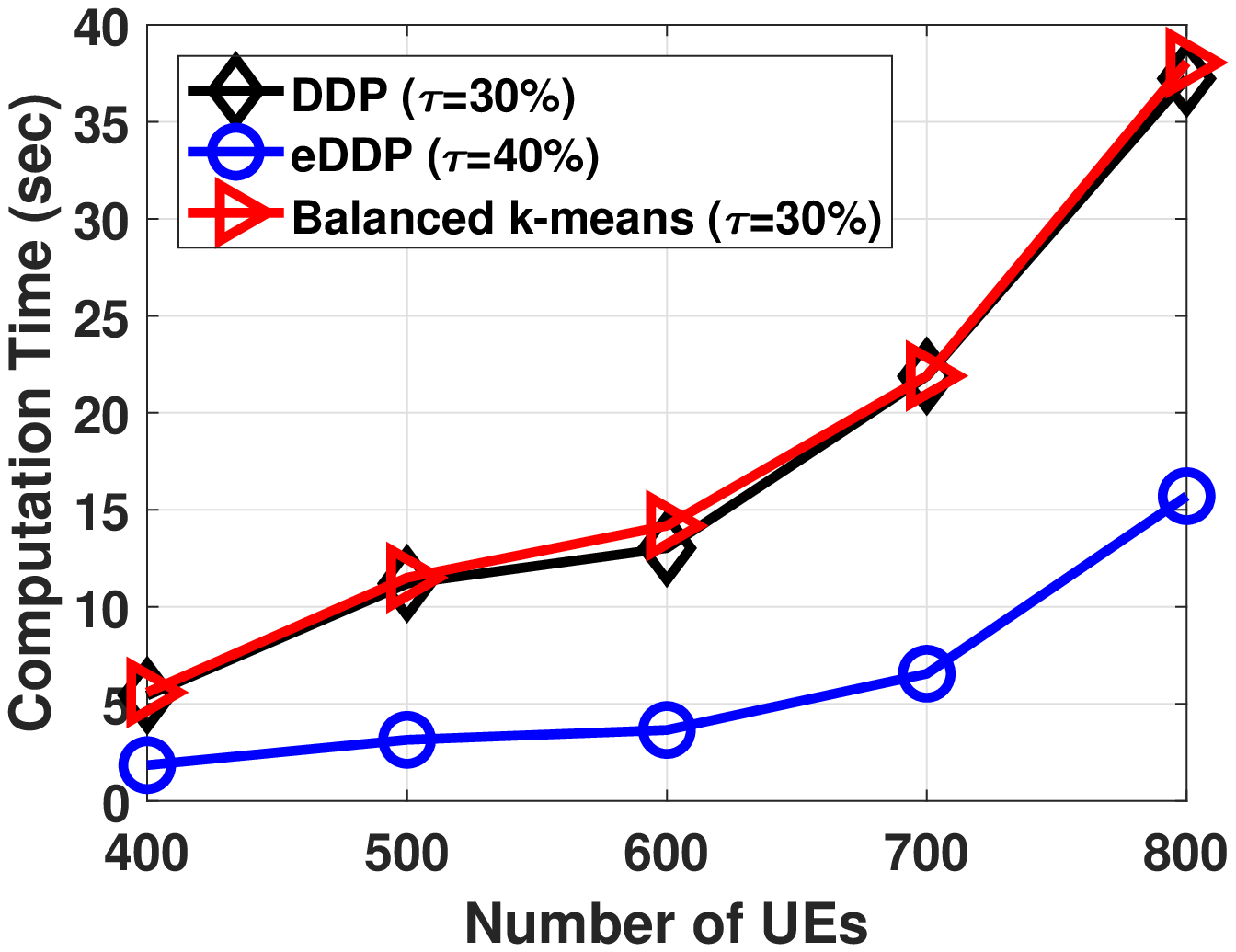}}%
	\subfigure[]{
		\label{fig:number_of_ues:number} 
		\includegraphics[width=0.33 \textwidth]{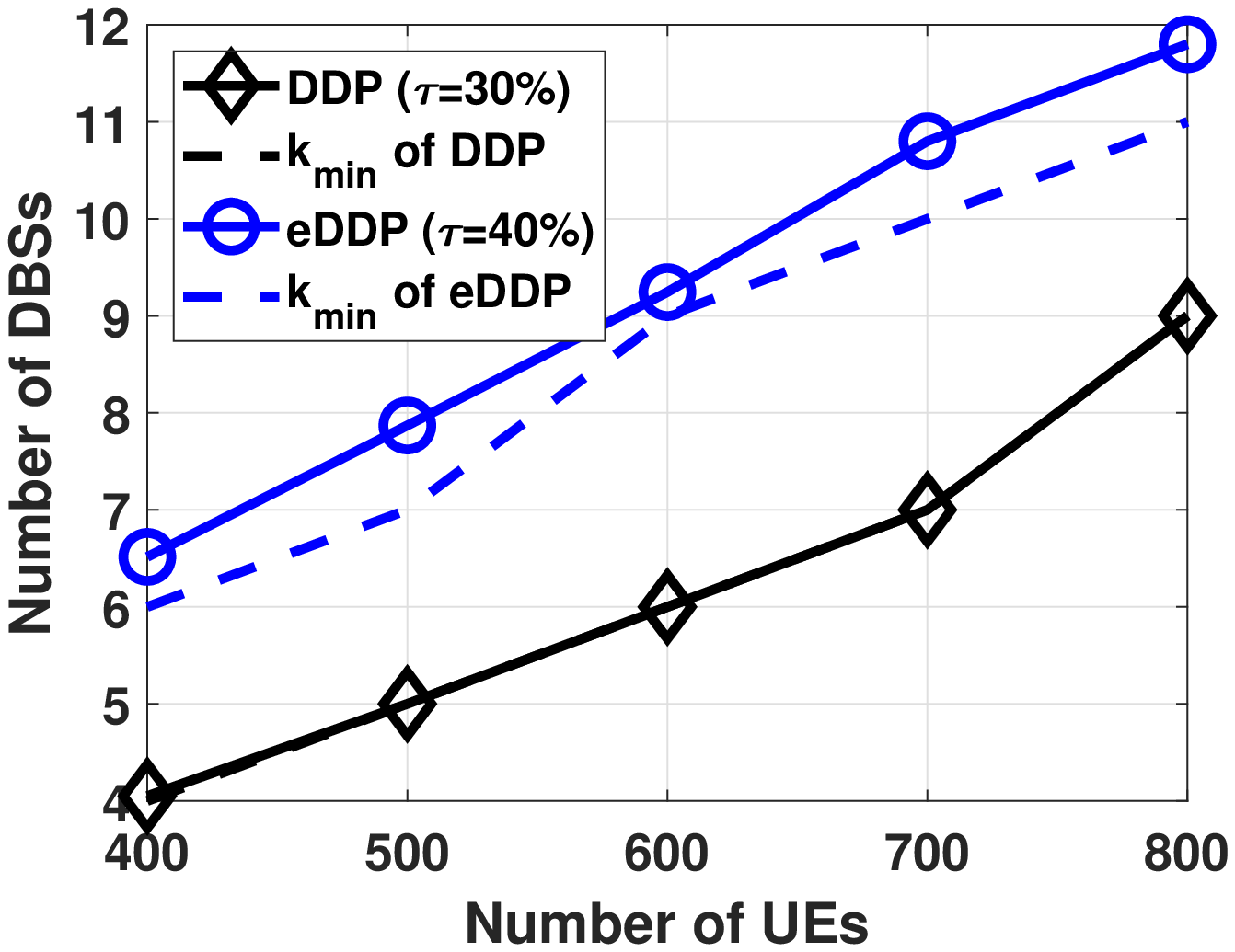}}%
	\subfigure[]{
		\label{fig:number_of_ues:sumrate} 
		\includegraphics[width=0.33 \textwidth]{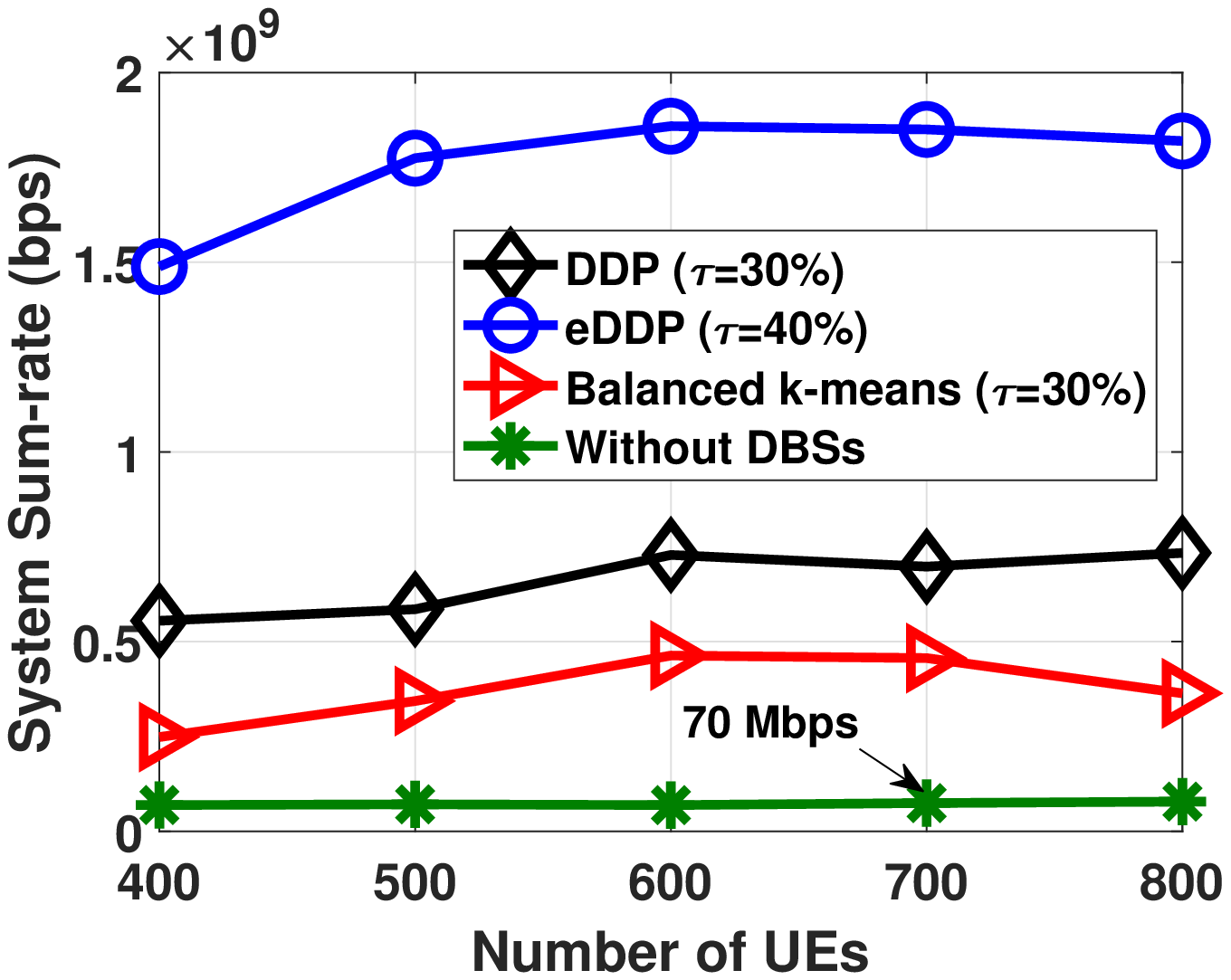}}%
	\caption{The performance results of different approaches in terms of~\subref{fig:number_of_ues:time} computation time,~\subref{fig:number_of_ues:number} number of DBSs, and~\subref{fig:number_of_ues:sumrate} system sum rate while serving different number of UEs with a given constraint $\tau$.
	}
	\label{fig:number_of_ues} 
\end{figure*}

We also have conducted some simulations with a given constraint, target satisfaction rate of UEs, to search a better placement. Fig.~\ref{fig:ratio_of_satisfied_ues} presents the performance results of DDP, eDDP, and balanced $k$-means approaches, in terms of number of DBSs, computation time, and system sum rate while varying the target satisfaction rate of UEs, $\tau$. 
We use the same input information as the ones in Fig.~\ref{fig:placement_result} and Fig.~\ref{fig:satisfaction_rate_of_methods} to evaluate every approach $1,000$ rounds to get the average results of different given target satisfaction rates. Comparing to DDP, eDDP will deploy a few more DBSs due to the pre-partition. According to the result in Fig.~\ref{fig:ratio_of_satisfied_ues}, eDDP is not sensitive to different target satisfaction rates in computation time. The reason is that our proposed Proposition~\ref{proposition1} can give a dynamic and tight lower bound of $k$ and it is proportional to the given target satisfaction rate as shown in Fig.~\ref{fig:ratio_of_satisfied_ues:number_of_dbss}. The gap between the exact number of DBS and the lower bound, $\Delta_{k}$, is almost a constant, so the system only needs to check a constant number of placement recommendations. Under ideal circumstances without interference, the computation time is fixed to $\mathcal{O}(N^3)$ and it is already analyzed in Section~\ref{sec:analysis}. Fig.~\ref{fig:ratio_of_satisfied_ues:number_of_dbss} indicates that eDDP effectively reduces more than $50\%$ computation time of DDP with the help of pre-partition and parallel processing.

When $\tau$ becomes higher, the lower bound of $k$ will increase. For DDP and the balanced $k$-means, as shown in Fig.~\ref{fig:ratio_of_satisfied_ues:number_of_dbss}, their computation time will slightly increase since the larger number of DBSs, $k$, leads more serious co-channel interference. Thus, both DDP and the balanced $k$-means need to iteratively search and check more placement candidates than eDDP. Although the number of deployed DBSs for DDP and the balanced $k$-means increases as $\tau$ becomes higher, their performance in system sum rates almost do not increase. There are two possible reasons. First, the effectiveness of their placement recommendations can not utilize the backhaul capacity of each DBS. Secondly, there are too much co-channel interference. Even though we have given the searching space with the lower and higher bounds of $k$ in Propositions~\ref{proposition1} and~\ref{proposition2}, due to the serious co-channel interference lead by coverage overlapping problem, DDP and the balanced $k$-means still can not find any feasible placements to meet the given constraints $\tau\geq45\%$ and $\tau\geq35\%$, respectively. Finally, only the proposed eDDP can effectively reduce the occurrence of the coverage overlapping problem and minimize the coverage overlapping area, so the system sum rate can be increased by eDDP as the given target satisfaction rate goes on. As a result, eDDP is the best approach to improve system sum rate.

\begin{table*}[!t]
	\renewcommand{\arraystretch}{1.1}
	\caption{Summary of Comparisons}
	\label{summary_table}
	\centering
	\begin{tabular}{l|lll}
		\hline
		\textbf{Comparison item} & \textbf{eDDP} & \textbf{DDP} & \textbf{Balanced $k$-Means} \\
		\hline 
		Value of $k$ & \textbf{Adaptive} & Adaptive & Predefined \\
		Scenario scale / Number of UEs & \textbf{Adaptive} & Adaptive & Predefined \\
		Co-channel interference & \textbf{Minimum} & Small & Medium \\
		Occurrence of coverage overlapping problem & \textbf{Seldom} & Sometimes & Frequently\\
		\hline
	\end{tabular}
\end{table*}
%

\subsection{The Number of UEs}
In the last simulation, we observe the impact of the number of UEs on the computation time, the number of DBSs, and system sum rate. As shown in Fig.~\ref{fig:number_of_ues:time}, the computation time of the balanced $k$-means and DDP increases exponentially as the number of UEs increases. Since the balanced $k$-means can not determine the number of DBSs, we directly use the $k$ given by Proposition~\ref{proposition1} of DDP. We find that DDP and balanced $k$-means cost almost the same computation time. However, in practice, without Proposition~\ref{proposition1}, the balanced $k$-means needs to search the placement starting from the case of $k=1$.

In contrast, eDDP requires less computation time than DDP since eDDP pre-partitions the input location information into multiple sub-regions and processes them in parallel. Fig.~\ref{fig:number_of_ues:time} indicates eDDP reduces more than 50\% computation time compared to DDP. The relation between the number of DBSs and the number of UEs is presented in Fig.~\ref{fig:number_of_ues:number}. In general, the number of DBSs is directly proportional to the number of UEs. Due to the task of pre-partition, eDDP always requires more DBS than DDP when serving different numbers of UEs. 
We can also find that the proposed DDP and eDDP approaches using Proposition~\ref{proposition1} only needs to check a very small number of cases.

The result in Fig.~\ref{fig:number_of_ues:sumrate} indicates that the proposed approach can improve system sum rate significantly comparing with the one without using any DBSs while serving different numbers of UEs. If the system serves the UEs without any DBSs, the maximum achievable system sum rate is only about $70$ Mbps. According to the result in Fig.~\ref{fig:number_of_ues:sumrate}, as the number of UEs increases, system sum rate using the balanced $k$-means is only $460$ Mbps. In contrast, the proposed DDP and eDDP can achieve about $720$ Mbps and $1.85$ Gbps, respectively.
There is a big performance improvement between the proposed solutions and the balanced $k$-means because the conventional solution leads much more co-channel interference (or coverage overlapping problem) as shown in Fig.~\ref{fig:placement_result:bkm}. When the number of UEs is larger than $700$, the sum rate of the balanced $k$-means starts to decreasing because of the serious co-channel interference.
Although the coverage overlapping problem may still occur, DDP outperforms the balanced $k$-means by 56\% system sum rate due the proposed placement refinement. The simulation result shows that the proposed eDDP performs the best in the system sum rate since eDDP can effectively reduce the occurrence of the coverage overlapping problem. As shown in Fig.~\ref{fig:number_of_ues:sumrate}, eDDP outperforms DDP by more than 100\% in terms of system sum rate.


\subsection{Comparison Summary and Open Issues}
We summarize the comparisons in Table~\ref{summary_table}. The proposed DDP and eDDP are more adaptive in scenarios of different scales and the number of required DBSs can be dynamically determined by the input parameters. In contrast with DDP and eDDP, the cellular operators need to determine the number of required DBSs in advance if they use a conventional $k$-means based approach. In addition, DDP and eDDP can outperform the balanced $k$-means in terms of co-channel interference and coverage overlapping problem. With the help of our proposed placement refinement step, DDP has smaller coverage overlapping area/co-channel interference than the balanced $k$-means approach. Combining the pre-partition step, eDDP can effectively alleviate the coverage overlapping problem and then provide the DBS placement with the minimum co-channel interference. As a result, the system sum rate is significantly improved.

The most important contribution of this work is to identify the coverage overlapping problem while deploying an on-demand multi-drone 3D cellular network for assisting a GBS to serve arbitrarily crowds. For the extension case of deploying multiple DBSs and coexisting with multiple GBSs, it is different from the problems considered in this work and is more complicated. For such a extension case, we briefly described some open issues as follows.
\begin{itemize}
	\item For a DBS, the interference from different GBSs are difficult to identify and mitigate.		
	\item For a GBS, a DBS is treated as a cellular connected UE. Until now, there is no common standard or solution to differentiate the aerial and ground UEs.
	\item New distributed mechanism/service/platform is required for different edge servers/GBSs to collaboratively manage the deployed DBSs.
	\item GBS and DBS need new mobility management technologies to coordinate the management of UE movement in heterogeneous coverage cells.
\end{itemize}
All of the above open issues are the emerging research topics in communications society. In summary, considering the case of multiple GBSs is not just a straightforward extension of this work and we consider this case as one of good open research topics in the future.

\section{Conclusion}
\label{sec:conclusion}
In this paper, we proposed a data-driven 3D placement (DDP) algorithm for deploying multiple DBSs to serve arbitrary flash crowds. DDP can automatically find the appropriate number, location, altitude, and coverage of DBSs in polynomial time. We also discover a new problem, coverage overlapping problem, which may frequently occur when the system places multiple DBSs to collaborate with the GBS to serve an arbitrarily distributed crowds. We then proposed an algorithm, enhanced DDP (eDDP), for solving the coverage overlapping problem. According to the simulation results, eDDP can provide a much more effective placement of higher sanctification rates than the balanced $k$-means and DDP. By using the suggested search space, including the upper and lower bounds of $k$, pre-partitioning and parallelism, eDDP can reduce the computation time by more than 50\% compared with DDP. In addition, eDDP can effectively reduce the occurrence of the coverage overlapping problem caused by the different altitudes of DBSs. As a result, compared to balanced $k$-means, eDDP effectively improves the system sum rate by more than 200\% and guarantees that more than 40\% UEs achieve the minimum data rate requirement.


\section*{Acknowledgment}
The authors would like to thank Professor Chun-Hung Liu for providing mathematical guidance. We are also very grateful to the editors and reviewers for their constructive suggestions on earlier versions of this article.

\bibliographystyle{IEEEtran}
\bibliography{IEEEabrv,reference}

\begin{IEEEbiography}[{\includegraphics[width=1in,height=1.25in,clip,keepaspectratio]{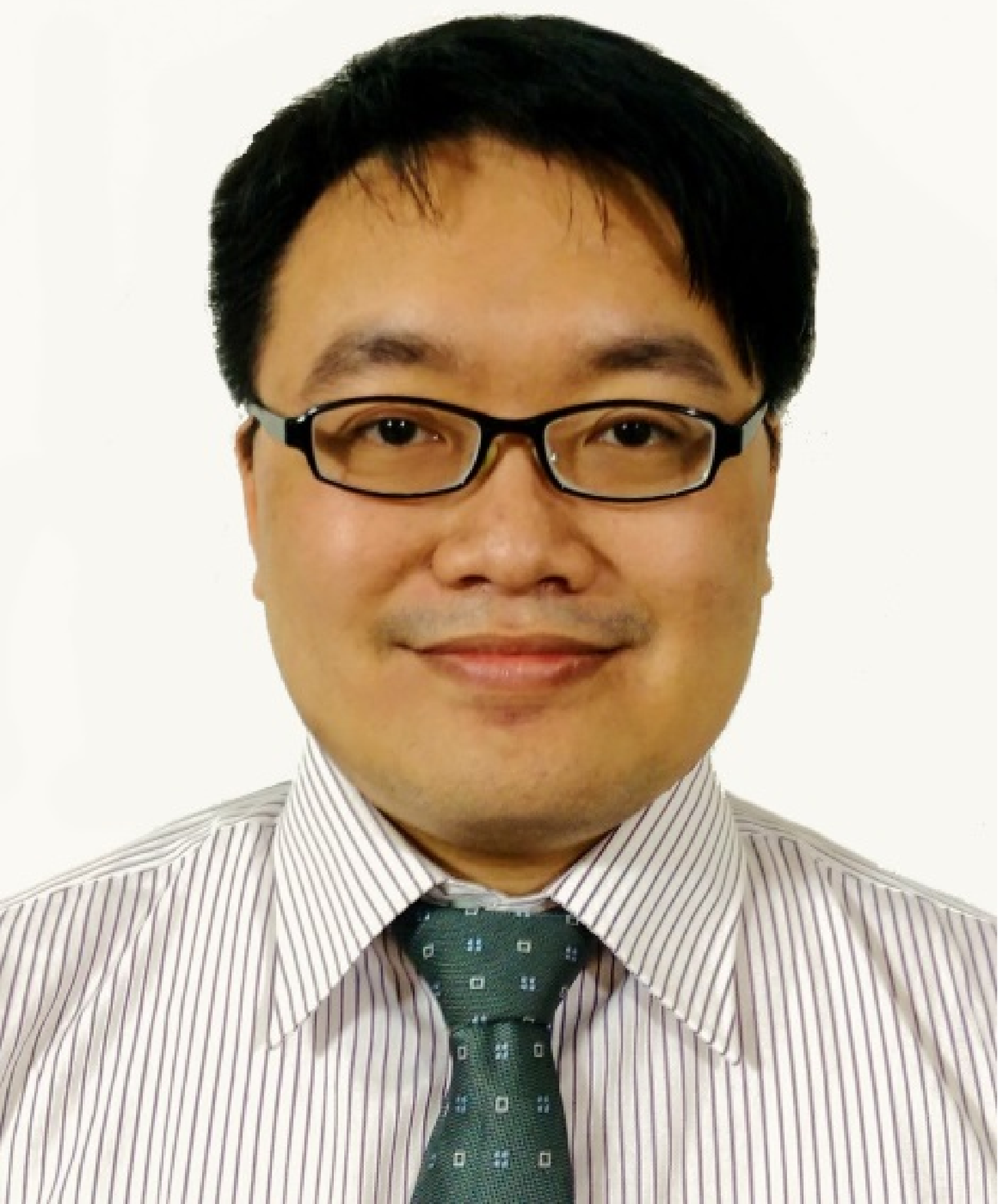}}]{Chuan-Chi Lai}
	(S'13 -- M'18) received the Ph.D. degree in computer science and information engineering from the National Taipei University of Technology, Taipei, Taiwan, in 2017.
	
	He is currently an assistant research fellow with the Department of Electrical and Computer Engineering, National Chiao Tung University, Hsinchu, Taiwan. His current research interests include resource allocation, data management, information dissemination techniques, and distributed query processing over moving objects in emerging applications such as the Internet of Things, edge computing, mobile wireless applications, and location-based services.
	
	Dr. Lai has received the Postdoctoral Researcher Academic Research Award of Ministry of Science and Technology, Taiwan, in 2019, the Best Paper Award in WOCC 2018 conference, and the Excellent Paper Award in ICUFN 2015 conference.
\end{IEEEbiography}

\begin{IEEEbiography}[{\includegraphics[width=1in,height=1.25in,clip,keepaspectratio]{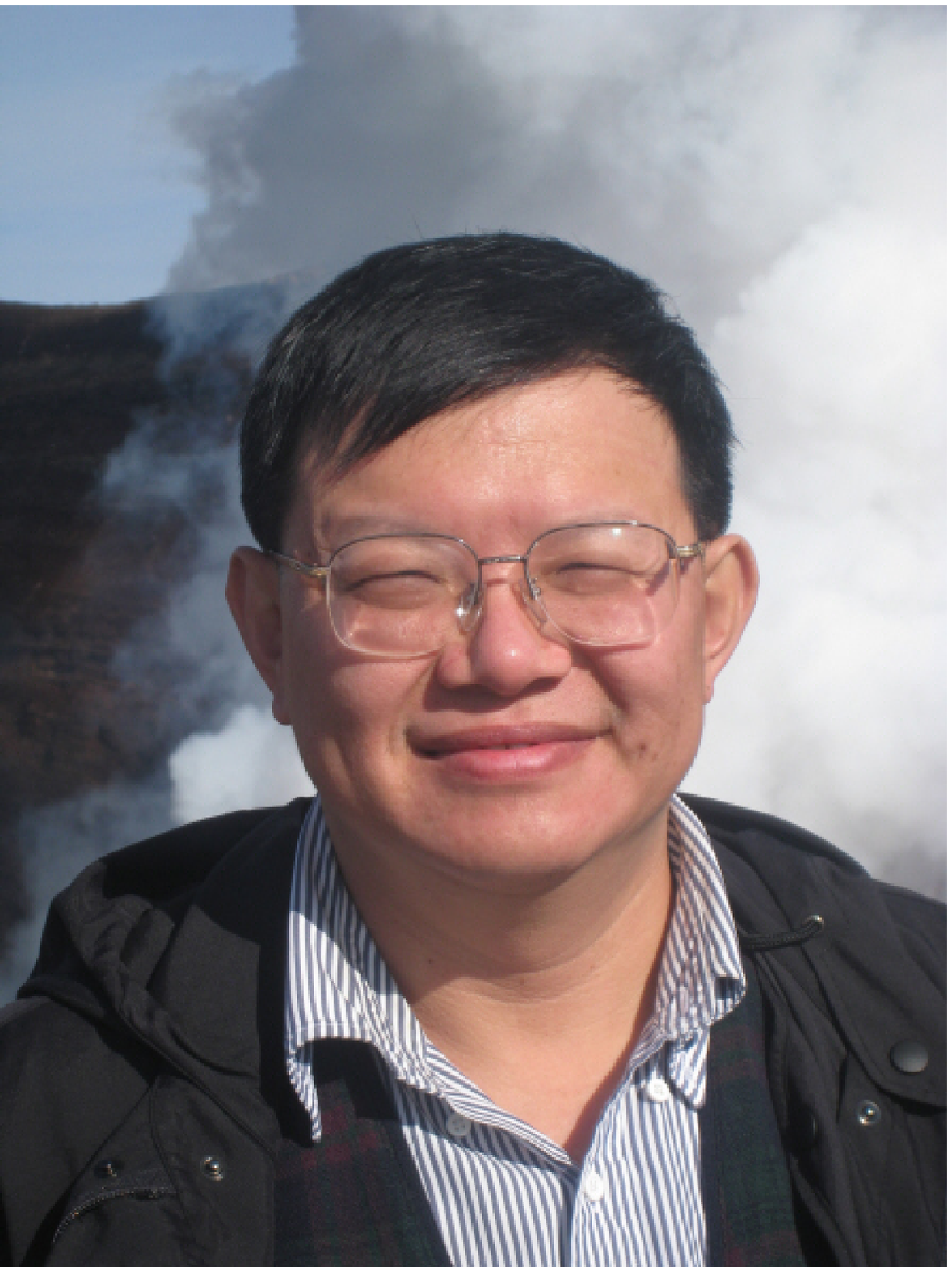}}]{Li-Chun Wang}
	(M'96 -- SM'06 -- F'11) received the Ph. D. degree from the Georgia Institute of Technology, GA, USA, in 1996. 
	
	From 1996 to 2000, he was with AT\&T Laboratories, Florham Park, NJ, USA, where he was a Senior Technical Staff Member in the Wireless Communications Research Department. Since 2000, he has been with the Department of Electrical and Computer Engineering of National Chiao Tung University, Hsinchu, Taiwan, where he is currently with the Department of Computer Science and Information Engineering. He holds 19 US patents, and has published over 200 journal and conference papers, and co-edited a book, "Key Technologies for 5G Wireless Systems," (Cambridge University Press 2017). His current research interests include software-defined mobile networks, heterogeneous networks, and data-driven intelligent wireless communications.
	
	Dr. Wang was a recipient of two Distinguished Research Awards of the National Science Council, Taiwan, in 2012 and 2017, and a co-recipient of the IEEE Communications Society Asia-Pacific Board Best Award in 2015, the Y. Z. Hsu Scientific Paper Award in 2013, and the IEEE Jack Neubauer Best Paper Award in 1997.	
\end{IEEEbiography}

\begin{IEEEbiography}[{\includegraphics[width=1in,height=1.25in,clip,keepaspectratio]{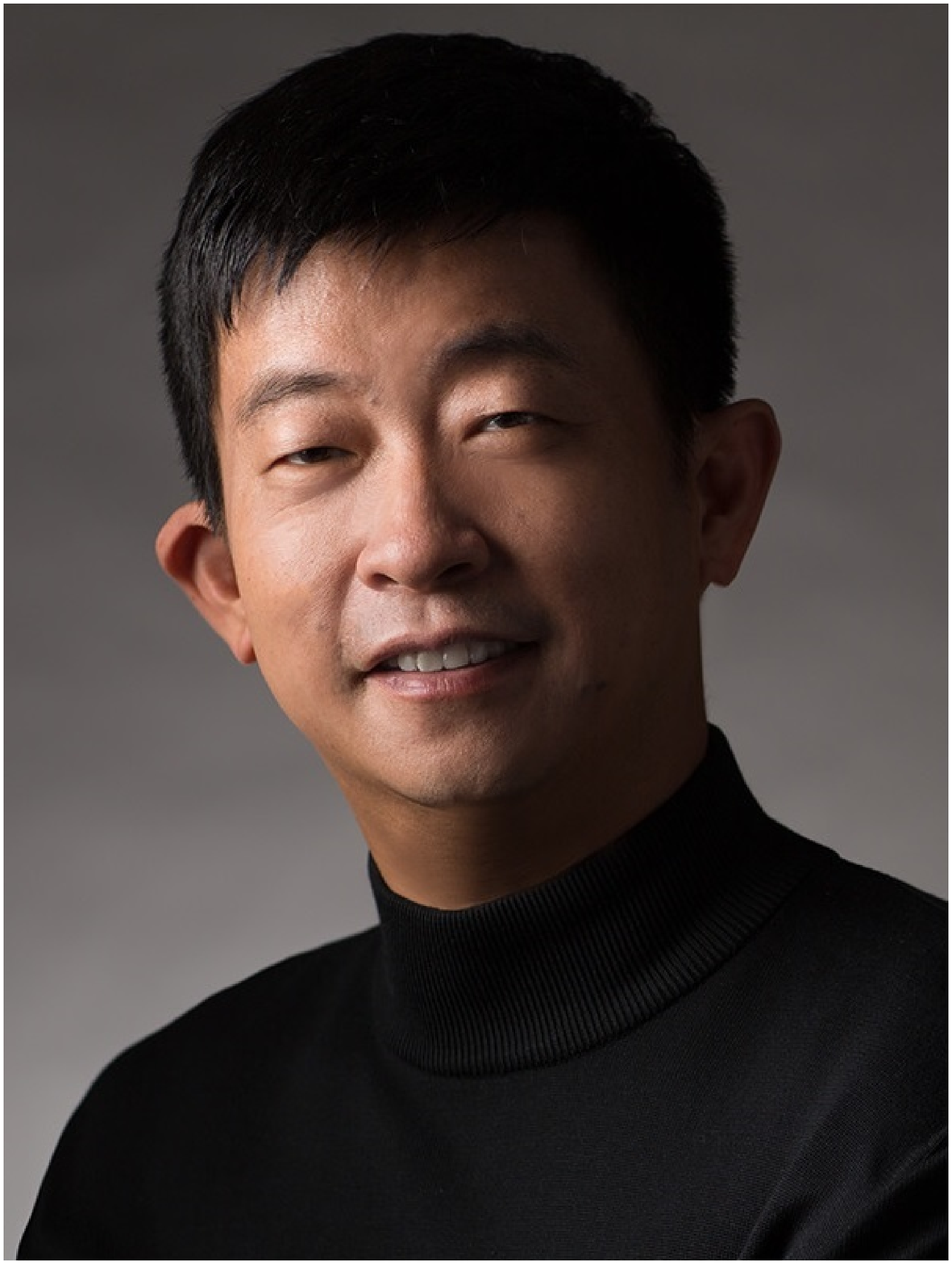}}]{Zhu Han}
	(S'01--M'04--SM'09--F'14) received the B.S. degree in electronic engineering from Tsinghua University, in 1997, and the M.S. and Ph.D. degrees in electrical and computer engineering from the University of Maryland, College Park, in 1999 and 2003, respectively. 
	
	From 2000 to 2002, he was an R\&D Engineer of JDSU, Germantown, Maryland. From 2003 to 2006, he was a Research Associate at the University of Maryland. From 2006 to 2008, he was an assistant professor at Boise State University, Idaho. Currently, he is a John and Rebecca Moores Professor in the Electrical and Computer Engineering Department as well as in the Computer Science Department at the University of Houston, Texas. His research interests include wireless resource allocation and management, wireless communications and networking, game theory, big data analysis, security, and smart grid. Dr. Han received an NSF Career Award in 2010, the Fred W. Ellersick Prize of the IEEE Communication Society in 2011, the EURASIP Best Paper Award for the Journal on Advances in Signal Processing in 2015, IEEE Leonard G. Abraham Prize in the field of Communications Systems (best paper award in IEEE JSAC) in 2016, and several best paper awards in IEEE conferences. Dr. Han was an IEEE Communications Society Distinguished Lecturer from 2015-2018, AAAS fellow since 2019 and ACM distinguished Member since 2019. Dr. Han is 1\% highly cited researcher since 2017 according to Web of Science. Dr. Han is also the winner of 2021 IEEE Kiyo Tomiyasu Award, for outstanding early to mid-career contributions to technologies holding the promise of innovative applications, with the following citation: ``for contributions to game theory and distributed management of autonomous communication networks."
\end{IEEEbiography}

	
	

\end{document}